\documentclass[amsthm]{elsart}

\usepackage{yjsco}
\usepackage{amsmath}
\usepackage{amsthm}
\usepackage{url}
\usepackage{epsfig}
\usepackage{amssymb}
\usepackage{algorithmic}
\usepackage{algorithm}
\usepackage{color}
\usepackage{relsize}
\usepackage{tikz}

\newcommand{\htmladdnormallink}[2]{#1}

\newtheorem{lemma}{Lemma}
\newtheorem{theorem}{Theorem}
\newtheorem{corollary}{Corollary}
\newcommand{\res}{\mathrm{res}}
\newcommand{\mult}{\mathrm{mult}}

\newcommand{\lcf}{\mathrm{lcf}}
\newcommand{\abs}[1]{\left| #1 \right|}
\newcommand{\norm}[1]{\left\| #1 \right\|}
\newcommand{\onenorm}[1]{\norm{#1}}
\newcommand{\tz}{\tilde{z}}
\newcommand{\hz}{\hat{z}}
\newcommand{\tp}{\tilde{p}}

\newcommand{\hp}{\hat{p}}

\newcommand{\tO}{\tilde{O}}

\newcommand{\sigmaai}{\sigma_{\alpha,i}}
\newcommand{\mai}{m_{\alpha,i}}
\newcommand{\taua}{\tau_{\alpha}}

\newcommand{\ma}{m(\alpha)}
\newcommand{\Pai}{P_{\alpha,i}}

\newcommand{\fa}{f_{\alpha}}
\newcommand{\ka}{k(\alpha)}
\newcommand{\tsigma}{\tilde{\sigma}}
\newcommand{\Z}{\mathbb{Z}}
\newcommand{\R}{\mathbb{R}}
\newcommand{\C}{\mathbb{C}}
\newcommand{\N}{\mathbb{N}}
\newcommand{\topo}{\textsc{Top}NT}
\newcommand{\IBox}{\mathfrak{B}}
\newcommand{\ignore}[1]{}
\newcommand{\ceil}[1]{\lceil #1 \rceil}
\newcommand{\floor}[1]{\lfloor #1 \rfloor}
\newcommand{\Mea}{\operatorname{Mea}}
\newcommand{\bd}{\mathop{\mathrm{bd}}}
\newcommand{\lc}{\operatorname{lcf}}
\newcommand{\sr}{\mathrm{sr}}

\usepackage{mathptmx}
\usepackage[scaled]{helvet}   
\usepackage{courier}

\usepackage{paralist}
\usepackage{psfrag}

\newcommand{\michael}[1]{\textcolor{red}{\textsc{Michael says: } {\sf #1}}}
\newcommand{\pengming}[1]{\textcolor{blue}{\textsc{Pengming: }{\sf #1}}}
\newcommand{\Kurt}[1]{#1}

\newcommand{\KurtTwo}[2]{#2}

\begin{document}
\begin{frontmatter}

\title{From Approximate Factorization to Root Isolation with
Application to Cylindrical Algebraic Decomposition}

\author{Kurt Mehlhorn}
\address{Max Planck Institute for Informatics, Saarbr\"ucken, Germany}
\ead{mehlhorn@mpi-inf.mpg.de}

\author{Michael Sagraloff}
\address{Max Planck Institute for Informatics, Saarbr\"ucken, Germany}
\ead{msagralo@mpi-inf.mpg.de}

\author{Pengming Wang}
\address{Max Planck Institute for Informatics, Saarbr\"ucken, Germany}
\ead{s9pewang@stud.uni-saarland.de}

\begin{abstract}
We present an algorithm for isolating all roots of an
  arbitrary complex polynomial $p$ that also works
  in the presence of multiple roots provided that \begin{inparaenum}[1)]
\item[(1)] the number of distinct roots is given as part of the input and
\item[(2)] the algorithm can ask for arbitrarily good approximations of the coefficients of $p$.
\end{inparaenum} The algorithm outputs pairwise disjoint disks each containing one of the
distinct roots of $p$ and the multiplicity of the root contained in the disk.
 The algorithm uses approximate factorization as a subroutine. For the case where Pan's algorithm~\cite{Pan:alg} is used for the factorization, we derive complexity bounds for the problems of isolating and refining all roots, which are stated in terms of the geometric locations of the roots only. Specializing the latter bounds to a
polynomial of degree $d$ and with integer coefficients of bitsize less than $\tau$, we show that $\tilde{O}(d^{3}+d^{2}\tau+d\kappa)$ bit operations are sufficient to compute isolating disks of size less than $2^{-\kappa}$ for all roots of $p$, where $\kappa$ is an arbitrary positive integer.

  In addition, we apply our root isolation algorithm to a recent algorithm for computing
  the topology of a real planar algebraic curve specified as the zero set of a bivariate
  integer polynomial and for isolating the real solutions of a bivariate
  polynomial system. For polynomials of degree $n$ and bitsize $\tau$, we 
  improve the currently best running
  time from
$\tO(n^{9}\tau+n^{8}\tau^{2})$ (deterministic) to $\tO(n^{6}+n^{5}\tau)$
(randomized) for topology computation and from $\tO(n^{8}+n^{7}\tau)$
(deterministic) to $\tO(n^{6}+n^{5}\tau)$ (randomized) for solving bivariate
systems.
\end{abstract}

\begin{keyword}
root isolation, root refinement, curve analysis, bivariate polynomial system, complexity analysis, cylindrical algebraic decomposition
\end{keyword}

\end{frontmatter}

\section{Introduction}
%
Root isolation is a fundamental problem of computational algebra \Kurt{and numerical analysis~\cite{Henrici64,Henrici74,McNamee-Pan,Mignotte,Yap:book}.}
Given a univariate polynomial $p$ with complex coefficients and possibly
multiple roots, 
the goal is to compute disjoint disks in the complex plane such that each disk contains exactly one root and the union of all disks covers all roots. We assume the existence of an oracle that can
be asked for rational approximations of the coefficients of arbitrary
precision. \Kurt{In particular, coefficients may be transcendental.} Note that non-rational coefficients can never be learned
exactly in finite time. 

In this generality, the problem is unsolvable. \Kurt{This is a consequence of the numerical halting problem~\cite{Yap-numerical-computation,Ker-I-Ko}. We give an example of a polynomial of degree three, for which, in the input model above, no finite algorithm can distinguish between the case of two or three distinct roots. 
When the coefficient oracle is asked for coefficients with precision $L$, it returns $1$, $\alpha_L$, $\beta_L$, and $2$, 
where $\alpha_L$ and $\beta_L$ are rational, the polynomial $p_L(x) = x^{3} + \alpha_L x^2 +
\beta_L x + 2$ has three distinct roots, and $\abs{(-2 \sqrt{2}-1) - \alpha_L} \le 2^{-L}$ and $\abs{\beta_L - (2 + 2 \sqrt{2})} \le 2^{-L}$. Observe that these answers of the oracle are consistent with the polynomials $p_L$ and $p$, where
$p(x) = (x - \sqrt{2})^2(x + 1) = x^3  + (-2 \sqrt{2}-1) x^2 +(2 + 2 \sqrt{2})x + 2$.
The former polynomial has three distinct roots and the latter polynomial has two distinct roots. Assume that the algorithm stops after asking for coefficients with precision $L$. If it outputs ``two distinct roots'', the oracle can claim that the input polynomial is $p_L$, if it outputs ``three distinct roots'', the oracle can claim that the input polynomial is $p$. In either case, the output is incorrect. }

The example shows that the problem needs to be restricted.
In addition to our assumption that the coefficients of our input polynomial $p$ are provided by coefficient oracles, we further assume that \emph{the number $k$ of
    distinct roots} is also given.\footnote{\Kurt{An alternative restriction is to be content with the computation of well-separated clusters of roots, i.e., the computation of disks  $\Delta_{i}$ and multiplicities $m_i$ such that $D_i$ contains exactly $m_i$ roots counted with multiplicity, $\sum_i m_i$ is equal to the degree of the polynomial, and substantially enlarged disks are disjoint. Our algorithm also applies to this version of the problem. We come back to it in Section~\ref{sec: well-separated clusters}.}} \KurtTwo{We would like to explain why this a reasonable assumption:}{The computation of $k$ requires symbolic methods. We would like to explain why knowledge of $k$ is nevertheless a reasonable assumption:} Root isolation is a key ingredient in the computation of a CAD (cylindrical algebraic decomposition) for a set of multivariate polynomials and, in particular, for computing the topology of algebraic curves and surfaces. In these applications, one has to deal with polynomials with multiple roots and algebraic coefficients; the coefficients are easily approximated to an arbitrary precision. In addition, the number of distinct roots is readily available from an algebraic precomputation (e.g.~computation of a subresultant sequence, triangular decomposition).
We now give an overview of our algorithm, our results, \Kurt{and related work}. 

\paragraph*{Root Isolation:} We fix some definitions which are used throughout the presentation: Let $p(x) = \sum_{i=0}^n p_ix^i \in \C[x]$, with $1/4\le |p_n|\le 1$,\footnote{The additional requirement for the leading coefficient $p_n$ yields a simpler presentation. Notice that, for general values $p_n$, we first have to multiply the polynomial $p$ by some $2^t$, with $t\in\Z$, such that $2^t\cdot |p_n|$ is contained in $[1/4,1]$.} be a complex polynomial with $k$ distinct roots $z_1,\ldots,z_k$. For $i=1,\ldots,k$, let $m_i:=\operatorname{mult}(z_{i},p)$ be the \emph{multiplicity} of $z_i$, and let $\sigma_i :=\sigma(z_{i},p):=
\min_{j \not= i} \abs{z_i - z_j}$ be the 
\emph{separation} of $z_i$ from the other roots of $p$. Then, our algorithm outputs
isolating disks $\Delta_{i}=\Delta(\tz_i,R_i)$ for the roots $z_{i}$ and
the corresponding multiplicities $m_i$. The radii satisfy 
$R_{i}<\frac{\sigma_{i}}{64n}$, and hence the center $\tilde{z}_{i}$ of
$\Delta_{i}$ approximates $z_{i}$ to an error of less than $\frac{\sigma_{i}}{64n}$. If the number of distinct roots of $p$ differs from $k$, we make no claims about termination and output.

The coefficients of $p$ are provided by 
oracles. That is, on input $L$, the oracle essentially returns
binary fraction approximations $\tilde{p}_{i}$ of the coefficients $p_{i}$ such that $\onenorm{p-\sum_{i=0}^{n} \tp_{i}x^{i}}\le 2^{-L}\onenorm{p}$. Here, $\norm{p}:=\norm{p}_{1}=|p_{0}|+\ldots+|p_{n}|$ denotes the \emph{one-norm} of $p$. The
details are given in Section~\ref{sec:setting}. \Kurt{The assumption that the coefficients are given through oracles is standard in computational real analysis~\cite{Ker-I-Ko} and numerical analysis~\cite{Henrici64,Henrici74,McNamee-Pan}, and is used in previous papers on approximate factorization and root isolation~\cite{Schoenhage,Pan:alg}}.

\Kurt{Many algorithms for approximate factorization and root isolation are known, see~\cite{Pan:survey} for a  survey. The algorithms can be roughly split into two groups: there are iterative methods for simultaneously approximating all roots (or a  single root if a sufficiently good approximation is already known), and there are subdivision methods that start with a region containing all the roots of interest, subdivide this region according to certain rules, and use inclusion- and exclusion-predicates to certify that a region contains exactly one root or no root. Prominent examples of the former group are the Aberth-Ehrlich method (used for MPSOLVE~\cite{Bini-Fiorentino}) and the Weierstrass-Durand-Kerner method. These algorithms work well in practice and are widely used. However, a complexity analysis and global convergence proof is missing. Prominent examples of the second group for isolating all complex roots are the Bolzano method~\cite{DBLP:journals/jsc/BurrK12,Yap-Sagraloff} and the splitting circle method~\cite{Schoenhage,Pan:alg}. There are also methods, e.g., the Descartes, Sturm, and continued fraction methods, for isolating the real roots of a real polynomial. Among the subdivision methods, the splitting circle method is asymptotically the best. It was introduced by Sch\"{o}nhage~\cite{Schoenhage} and later considerably refined by Pan~\cite{Pan:alg}. An implementation of the splitting circle method in the computer algebra system Pari/GP is available~\cite{Gourdon}. None of the algorithms mentioned deals specifically with multiple roots. For square-free polynomials, i.e, the case $k = n$, the subdivision methods guarantee root isolation. For integral polynomials, a square-free decomposition can be computed~\cite{Gathen-Gerhard:book}. Alternatively, separation bounds~\cite[Section 6.7]{Yap:book} can be used to guarantee isolation in the presence of multiple roots. Johnson~\cite{Johnson91}, Cheng~et.~al.~\cite{Cheng-Gao-Yap09}, and Strzebonski and Tsigaridas~\cite{Strzebonski-Tsigaridas11} discuss root isolation for polynomials with algebraic coefficients.}

\Kurt{Strzebonski~\cite{Strzebonski} presents an algorithm that deals with multiple roots in our setting. However, it has heuristic steps. The algorithm in~\cite{Mehlhorn-Sagraloff} can cope with at most one multiple
root and needs to know the number of distinct complex roots as well as the
number of distinct real roots.  Algorithms for root refinement, e.g., Newton-Raphson iteration, compute arbitrary good approximations to roots once a good initial approximation is known. Generalization to clusters of roots are provided by~\cite{Yakoubsohn00,Giusti-Lecerf-Salvy-Yakoubsohn05}.}

Our algorithm has a simple structure. \Kurt{It combines mainly known techniques. Our contribution is the right assembly into an algorithm, our novel clustering step, and the complexity analysis.} We first use any algorithm (e.g.~\cite{Bini-Fiorentino,Schoenhage,Pan:alg,Yap-Sagraloff}) for
approximately factorizing the input polynomial.  It is required that it can be run
with different levels of precision, and that, for any given integer $b$, it returns
approximations $\hz_1$ to $\hz_n$ for the roots of $p$ such that
\begin{align}\label{intro:factorization}   \norm{p -  p_n \prod\nolimits_{1 \le j \le n} (x- \hz_j)}\le 2^{-b}\norm{p}. \end{align}
In a second step, we partition
the root approximations $\hz_1$ to $\hz_n$ into \KurtTwo{$k$ clusters $C_{1},\ldots,C_k$ based on geometric
vicinity}{clusters $C_{1},C_2\ldots $ based on geometric vicinity. If the number of clusters is less than $k$, we increase the precision, and refactor. The difficulty of the clustering step lies in the fact that the amounts by which roots will move after a perturbation of the coefficients (recall that, in our input model, we only see perturbations of the true coefficients) depends heavily on the multiplicity of the root.} We enclose each cluster $C_{i}$ in a disk $D_{i}=\Delta(\tilde{z}_{i},r_{i})$ and make sure that the disks are pairwise disjoint and that the radii $r_{i}$ are not ``too small'' compared to the pairwise distances of the centers $\tz_{i}$.\footnote{This is crucial to control the cost for the final verification step. For details, we refer to Sections~\ref{sec:clustering} and \ref{sec:certification}.} In a third step, we verify that the $n$-times enlarged disks
$\Delta_{i}=\Delta(\tilde{z}_{i},R_{i})=\Delta(\tilde{z}_{i},n\cdot r_{i})$ are disjoint and that each of them
contains exactly
the same number of approximations as roots of $p$ counted with
multiplicity. \Kurt{As in~\cite{Yakoubsohn00,Giusti-Lecerf-Salvy-Yakoubsohn05}, we use Rouch\'{e}'s theorem for the verification step.}
If the clustering and the verification succeed, we return the disks $\Delta_1,\ldots,\Delta_k$
and the number of approximations $\hz\in\{\hz_{1},\ldots,\hz_{n}\}$ in the disk as
the multiplicity of the root isolated by the disk. If either clustering or verification does not
succeed, we repeat with a higher precision. Strzebonski~\cite{Strzebonski} has
previously described a similar approach. The main difference is that he used a
heuristic for the clustering step and hence could neither prove completeness of
his approach nor analyze its complexity. He reports that his algorithm does
very well in the context of CAD computation. 

In the example above, we would have the additional information that $p$
  has exactly two distinct roots. We ask the oracle for an $L$-approximation of
  $p$ for sufficiently large $L$ and approximately factor it. Suppose that we obtain
 approximations $-1.01$, $1.4$, and $1.42$ of the roots, and let $\hp = (x + 1.01)(x - 1.4)(x - 1.42)$.
 The clustering step may then put  the first approximation
  into a singleton cluster and the other two approximations into a cluster of
  size two. It also computes disjoint enclosing disks. The verification step
  tries to 
 certify that $p$ and $\hp$ contain the same number of roots in both
 disks. If $L$ and $b$ are sufficiently large, clustering and verification succeed.

If Pan's algorithm~\cite{Pan:alg} is used for the approximate
factorization step, then the overall algorithm has bit complexity\footnote{$\tilde{O}$ indicates that we omit logarithmic factors.}
\begin{align}\label{intro:result1}
 		{\mathlarger{ \tilde{O}}}\left(n^{3}+n^{2}\sum_{i=1}^{k}\log M(z_{i})+n\sum_{i=1}^{k} \log \left(M(\sigma_{i}^{-m_{i}})\cdot M(P_{i}^{-1})\right)\right)
		  \end{align}
where $P_{i}:=\prod_{j\neq i} (z_{i}-z_{j})^{m_{j}}=\frac{p^{(m_{i})}(z_{i})}{m_{i}!p_{n}}$, and $M(x):=\max(1,|x|)$.
Observe that our algorithm is adaptive in a very strong sense, namely, the above bound exclusively depends on the actual multiplicities and the geometry (i.e.~the actual modulus of the roots and their distances to each other) of the roots. 
There is also no dependency on
the size or the type (i.e.~whether they are rational, algebraic or
transcendental) of the coefficients of $p$. 

Our algorithm can also be used to further refine the isolating disks to a size
of $2^{-\kappa}$ or less, where $\kappa$ is a given integer. The bit complexity
for the refinement is given by the bound in (\ref{intro:result1}) plus an
additional term $\tilde{O}(n\cdot \kappa\cdot \max_i m_i )$. In particular, for square-free polynomials the amortized cost per root and bit of precision is
$\tO(1)$, and hence the method is optimal up to polylogarithmic factors.

For the benchmark problem of isolating all roots of a polynomial $p$ with \emph{integer} coefficients of absolute value bounded by
$2^{\tau}$, the bound in (\ref{intro:result1}) becomes
$\tO(n^3 + n^2 \tau)$.\footnote{We first divide $p$ by its leading coefficient to meet the requirement on the leading coefficient, and apply our algorithm to $p/p_n$.} The bound for the refinement becomes $\tO(n^3 + n^2 \tau+n\kappa)$, even if 
there exist multiple roots. 

For a square-free integer polynomial $p$, an algorithm by Pan~\cite[Theorem 3.1]{Pan:survey} achieves a comparable complexity bound for the benchmark problem. That is, based on the computations in~\cite[Section~20]{Schoenhage}, one can compute a bound $b_{0}$ of size $\Theta(n(\tau+\log n))$ with the property that if $n$ points $\hz_{j}\in\C$ fulfill the inequality (\ref{intro:factorization}) for a $b\ge b_{0}$, then they approximate the corresponding roots $z_{j}$ to an error less than $\sigma_{j}/(2n)$; cf. Lemma~\ref{lem:approximation quality} for an adaptive version. Hence, for $b\ge b_{0}$, Pan's factorization algorithm also yields isolating disks for the roots of $p$ using $\tO(n^{2}\tau)$ bit operations. Note while this approach achieves a good worst case complexity, however, it is for the price of running the factorization algorithm with $b=\Theta(n(\tau+\log n))$ even when the roots are well conditioned. In contrast, our algorithm turns Pan's factorization algorithm into a highly adaptive method for isolating and approximating the roots of a general polynomial.
Also, for general polynomials, there exist bounds~\cite[Section 19]{Schoenhage} for the distance between the roots of $p$ and corresponding approximations fulfilling (\ref{intro:factorization}). They are optimal for roots of multiplicity $\Omega(n)$ but overestimate badly if all roots have considerably smaller multiplicities. For the task of root refinement, the bit complexity of our method depends on $\kappa$ as $\tO(n\max_{i} m_{i}\cdot\kappa)$ and, hence, it adapts to the highest occurring multiplicity, whereas previous 
methods~\cite{qir-kerber-11,Pan:alg,Sagraloff12} depend as  $\tO(n^{2}\kappa)$.

\paragraph*{Topology Computation and Computing Real Solutions of Bivariate Systems}
Our new root isolation algorithms has an interesting consequence on the complexity of computing the topology (in terms of a cylindrical algebraic
decomposition) of a real planar algebraic curve specified
as the zero set of an integer polynomial and of isolating the real solutions of a
bivariate polynomial system. \Kurt{Both problems are well-studied~\cite{Collins:CAD,bpr-arag-06,beks:top2D,Bouzidi-et-al-2,Bouzidi-et-al-1,Cheng-et-al,Cheng-Jin-Lazard13,cjk-iaicad-02,Diochnos:2009,Eigenwillig-Kerber-Wolpert,es-bisolvecomplexity-11,Emiris-Tsigaridas05,Gonzales-Vega,Hong,ks-top-12,Rouillier-Survey-BivariateSystemSolving,Strzebonski}}.
\Kurt{The latter problem can be reduced to the former as the real solutions of the bivariate system $f(x,y) = g(x,y) = 0$ correspond to the points on the real curve $f^2(x,y) + g^2(x,y) = 0$.} In Section~\ref{sec:curveana}, we apply our method
to a recent algorithm \textsc{TopNT}~\cite{beks:top2D}
for computing the topology of a planar algebraic curve. This yields bounds on the \emph{expected} 
number of bit operations.
which improve the currently best (which are both deterministic)
bounds~\cite{es-bisolvecomplexity-11,ks-top-12} from
$\tO(n^{9}\tau+n^{8}\tau^{2})$ to $\tO(n^{6}+n^{5}\tau)$ for topology computation and from $\tO(n^{8}+n^{7}\tau)$
to $\tO(n^{6}+n^{5}\tau)$ for solving bivariate
systems. 

\KurtTwo{The improvement is due to a series of ideas. First, we exploit the fact that \textsc{TopNT} mainly uses resultant and gcd computations a purely symbolic operations which allows to derive low bit complexity for the complexity bounds. In addition, as numerical steps, \textsc{TopNT} uses root isolation and refinement as considered in the first part of this paper. In particular, we use our algorithm for computing the roots of the "fiber" polynomials $f(\alpha,y)\in\C[y]$, where $\alpha$ is an $x$-critical point of a planar algebraic curve $C=V(f)\subset\C^2$ and the number of distinct roots of $f(\alpha,y)$ is given from an algebraic pre computation.
Combining the adaptive complexity bounds from this paper and amortized complexity bounds from~\cite{ks-top-12} for all fiber polynomials $f(\alpha,y)$ eventually yields considerably improved complexity bounds for the numerical steps as well.}{As several other recent algorithms~\cite{Strzebonski,Eigenwillig-Kerber-Wolpert,Cheng-Gao-Yap09,Cheng-Jin-Lazard13} for topology computation or bivariate system solving, \textsc{TopNT} uses numerical computation as much as possible. In particular, the symbolic operations are restricted to resultant and $\gcd$ computations, which do not dominate the overall bit complexity. The workhorse in \textsc{TopNT} is root isolation and refinement as considered in the first part of this paper, in particular, the isolation of the roots of the "fiber" polynomials $f(\alpha,y)\in\C[y]$,  where $\alpha$ is an $x$-critical point of a planar algebraic curve defined as the vanishing set of a polynomial $f \in \Z[x,y]$. The number of distinct roots of $f(\alpha,y)$ is available from an algebraic precomputation.
Combining the adaptive complexity bounds from this paper and the amortized complexity bounds from~\cite{ks-top-12} for all fiber polynomials $f(\alpha,y)$ eventually yields considerably improved complexity bounds for the numerical steps.}
\bigskip


\paragraph*{Paper History:} An extended abstract~\cite{MSW13} of this paper was presented at ISSAC 2013. The current paper extends the conference version significantly. In particular, the analysis of the algorithm for root isolation (i.e.~the results in Section~\ref{sec:rootisolation})
was only sketched (Lemma~\ref{lem:tauTau} 
and Theorem~\ref{thm:rootbound} were stated without proof, and only a sketch of the proof of Theorem~\ref{thm:square free complexity} was given), and the application 
of our root isolation algorithm to curve topology computation and to solving bivariate polynomial systems as well as the corresponding analysis (i.e.~Section~\ref{sec:curveana}) was not covered at all in the extended abstract.

\newpage


\section{Root Finding}\label{sec:rootisolation}

\subsection{Setting and Basic Properties}\label{sec:setting}

We consider a polynomial
\begin{align}
p(x)=p_{n}x^{n}+\ldots+p_{0}\in\C[x]\label{def:poly}
\end{align}
of degree $n\ge 2$, where $1/4\le p_{n}\le 1$. We fix the following notations:\medskip
\begin{itemize}
\item $M(x):=\max(1, |x|)$, for $x\in \mathbb{R}$,
\item $\tau_{p}$ denotes the minimal non-negative integer with
  $\frac{|p_{i}|}{|p_n|}\le2^{\tau_{p}}$ for all $i=0,\ldots,n-1$, 
\item $\onenorm{p} :=\norm{p}_{1}:= \abs{p_0} + \ldots + \abs{p_n}$ denotes the $1$-norm of
  $p$,
\item $z_{1},\ldots,z_k\in\C$ are the distinct complex roots of $p$, with $k\le n$,
\item $m_{i}:=\operatorname{mult}(z_{i},p)$ is the multiplicity of $z_{i}$,
\item $\sigma_{i}:=\sigma(z_{i},p):=\min_{j\neq i}|z_{i}-z_{j}|$ is the \emph{separation} of $z_{i}$,
\item $\Gamma_p:=M(\max_i \log M(z_i))$ denotes the \emph{logarithmic root bound} of $p$, 
\item $\Mea(p)=|p_n|\cdot\prod_i M(z_i)^{m_{i}}$ denotes the \emph{Mahler Measure} of $p$. 
\end{itemize}

The quantities $\tau_p$, $\Gamma_p$, $|p_{n}|$ and $\Mea(p)$ are closely related.

\begin{lemma}\label{lem:tauTau} $\Gamma_p \le  1 + \tau_p$ and $\tau_p-n-1 \le  \log\frac{\Mea(p)}{|p_{n}|}\le n\Gamma_{p}$. \end{lemma}
\begin{proof} By Cauchy's root bound $\max_i \abs{z_i} \le 1+ \max_i
  \abs{p_i}/\abs{p_n}$, and thus $\max_i \log \abs{z_i} \le 1 + \tau_p$. Since
  $\tau_p \ge 0$, by definition, we have $\Gamma_p \le 1 + \tau_p$. 
The $i$-th coefficient of $p$ is smaller than or equal to
$\binom{n}{i}\Mea(p)\le 2^{n}\Mea(p)\le 2^{n(\Gamma_{p}+1)}$.\ignore{\michael{I went back to the old version. Note that the Mahler measure already contains the factor $|p_{n}|$}} Thus, from the definition of $\tau_{p}$,
either $\tau_{p}=0$ or
$2^{n}\frac{\Mea(p)}{|p_{n}|}\ge \max_i\frac{|p_i|}{|p_n|} \ge 2^{\tau_{p}-1}$
\end{proof}

We assume the existence of an oracle which provides arbitrary good approximations of the  polynomial $p$. Let $L \ge 1$ be an integer. 
We call a polynomial $\tp=\tp_{n}x^{n}+\ldots+\tp_{0}$, with
$\tp_{i}=s_{i}\cdot 2^{-(L+1)}$ and $s_{i} \in \Z$,  an
\emph{absolute $L$-approximation} of $p$ if $|\tp_{i}-p_{i}| \le 2^{-L}$. 
We further assume that we can ask for such an approximation $\tp$ for the cost $\tO(n(L+\tau_p))$. This is the cost of reading the coefficients of $\tp$.

We call a polynomial $\tp=\tp_{n}x^{n}+\ldots+\tp_{0}$, with
$\tp_{i}=s_{i}\cdot 2^{-(L+1)}$ and $s_{i} \in \Z$,  a
\emph{relative $L$-approximation} of $p$ if $\onenorm{\tp-p} \le 2^{-L}\onenorm{p}$. Since $L\ge 1$, the triangle inequality implies that
  \begin{align}\label{apx:normp}
  \frac{\onenorm{\tp}}{2} \le \onenorm{p} \le 2 \onenorm{\tp}.
  \end{align}
Furthermore, notice that any absolute $(L+\lceil\log (n+1)\rceil+2)$-approximation of $p$ is also a relative $L$-approximation of $p$ because of $\onenorm{\tp-p}\le (n+1)\cdot 2^{-L-\lceil\log (n+1)\rceil-2}\le |p_n|\cdot 2^{-L}\le \onenorm{p}\cdot 2^{-L}$. Hence, we can ask for a relative $L$-approximation for the cost $\tO(n(L+\tau_p))=\tO(n(L+\Mea(p)))$.

In the next step, we show that 
a ``good'' integer approximation $\Gamma$ of $\Gamma_{p}$ can be computed
with $\tO(n^2 \Gamma_p)$ bit operations.

\begin{theorem}\label{thm:rootbound}
An integer $\Gamma\in\N$ with
\begin{align}
\Gamma_p\le\Gamma<8\log n+\Gamma_p\label{inequ:rootbound}
\end{align}
can be computed with $\tO(n^{2}\Gamma_p)$ bit operations. The computation uses
an absolute $L$-approximation of precision $L$ of $p$ with $L=O(n\Gamma_{p})$.
\end{theorem}

\begin{proof}
\ignore{We first compute an approximation $\tp$ of precision $L$ of $p$ with $\abs{\tp_n}/2 \le \abs{p_n}
\le 2 \cdot \abs{\tp_n}$ as described above. Let 
$\kappa:=2^{\floor{\log \abs{\tp_{n}}} - 1}$. Then $\kappa \le \abs{\tp_n}/2 \le
  \abs{p_n}$ and $\kappa \ge 2^{\floor{\log \abs{p_{n}/2}} - 1 }\ge \abs{p_n}/8$.
Consider the scaled polynomial
  $q := p/\kappa$. Its leading coefficient lies between $1$ and $8$.  Moreover,
$\Gamma_q = \Gamma_{p}$, $\tau_q = \tau_{p}$, and an arbitrary approximation of precision $L$ of $p$ 
also yields an approximation of precision $L$ of $q$ since division by $\lambda$ is just a shift 
by $\log\kappa$ bits and $\onenorm{q} =\onenorm{p}/\kappa$.
Hence, we may assume that $1 \le \abs{p_n} \le 8$.} 
 
We consider the \emph{Cauchy polynomial}
$$\bar{p}(x):=|p_n|x^n-\sum_{i=0}^{n-1}|p_i|x^i$$ of $p$. Then, according
to~\cite[Thm. 8.1.4.]{rahman2002analytic} or~\cite[Thm. 3.8(e)]{vdSl70}, $\bar{p}$ has a unique positive
real root $\xi\in\R^+$, and the following inequality holds:
\[ \max_i \abs{z_i} \le \xi < \frac{n}{\ln 2}\cdot \max_i \abs{z_i} <2 n\cdot
  \max_i \abs{z_i}.\]
It follows that $\bar{p}(x)>0$ for all $x\ge \xi$ and $\bar{p}(x)< 0$ for all
$x< \xi$. Furthermore, since $\bar{p}$ coincides with its own Cauchy
polynomial, each complex root of $\bar{p}$ has absolute value less than or
equal to $|\xi|$.  
Let $k_{0}$ be the smallest non-negative integer 
$k$ with $\bar{p}(2^k)>0$ (which is equal to the smallest $k$ with $2^{k}>\xi$). Our goal is to 
compute an integer $\Gamma$ with $k_{0}\le\Gamma\le k_{0}+1$. Namely, if $\Gamma$ fulfills the 
latter inequality, then $M(\max_{i}|z_{i}|)\le M(\xi)\le 2^{\Gamma}<4 M(\xi)< 8n\cdot M(\max_{i}|z_{i}|)$, and thus $\Gamma$ fulfills inequality (\ref
{inequ:rootbound}). 
In order to compute a $\Gamma$ with $k_{0}\le\Gamma\le k_{0}+1$, we use
exponential and binary
search (try $k = 1, 2, 4, 8, \ldots$ until $\bar{p}(2^k) > 0$ and, then, perform
binary search on the interval $k/2$ to $k$) and approximate evaluation of $\bar{p}$ at the points $2^{k}$: More
precisely, we evaluate $\bar{p}(2^{k})$ using interval arithmetic with a
precision $\rho$ (using fixed point arithmetic) which guarantees that the width
$w$ of $\IBox(\bar{p}(2^k),\rho)$ is smaller than $1$,
where $\IBox(E,\rho)$ is the interval obtained by evaluating a polynomial
expression $E$ via interval arithmetic with precision $\rho$ for the basic
arithmetic operations; see~\cite[Section 4]{qir-kerber-11} for details. We
use~\cite[Lemma 3]{qir-kerber-11} to estimate the cost for each such
evaluation: Since $\bar{p}$ has coefficients of size less than
$2^{\tau_{p}}|p_{n}|<2^{\tau_{p}}$, we have to choose $\rho$ such that  
\[
2^{-\rho+2}(n+1)^2 2^{\tau_{p}+nk}<\frac{1}{4}
\] 
in order to ensure that $w < 1/4$. 
Hence, $\rho$ is bounded by $O(\tau_{p}+nk)$ and, 
thus, each interval evaluation needs $\tilde{O}(n(\tau_{p}+nk))$ bit
operations. We now use exponential plus binary search to find the  
smallest $k$ such that $\IBox(\bar{p}(2^k),\rho)$ contains only positive
values. The following argument  
then shows that $k_{0}\le k\le k_{0}+1$: Obviously, we must have $k\ge k_{0}$
since 
$\bar{p}(2^{k}) 
<0$ and $\bar{p}(2^{k})\in \IBox(\bar{p}(2^k),\rho)$ for 
all $k<k_{0}$. Furthermore, the point $x=2^{k_{0}+1}$ has distance more than $1$ to each of the roots 
of $\bar{p}$, and thus $|\bar{p}(2^{k_{0}+1})|\ge |p_{n}|\ge 1/4$. Hence, it
follows that 
$\IBox(\bar{p}(2^{k_0} + 1),\rho)$ contains only positive values. 
For the search, we need $$O(\log k_{0})=O(\log\log\xi)=O(\log(\log n+\Gamma_{p}))$$ 
iterations, and the cost for each of these iterations is bounded by $\tilde{O}(n(\tau_{p}+nk_{0}))=\tilde{O}(n^{2}\Gamma_{p})$ bit 
operations. \end{proof}


\subsection{Algorithm}\label{sec:algorithm}

We present an algorithm for isolating the roots of a polynomial
$p(x)=\sum_{i=0}^n p_ix^i=p_n\prod_{i=1}^k (x-z_i)^{m_i}$,  
where the coefficients $p_i$ are given as described in the previous section.
We may assume that $k > 1$; the problem is trivial otherwise. If $k = 1$, then
$-p_{n-1}/(n p_n)$ is the root of multiplicity $n$. 
The algorithm uses some polynomial factorization algorithm to
produce approximations for the roots $z_1,\ldots,z_k$,  
and then performs a clustering and certification step to verify that the
candidates are of high enough quality. For concreteness, we pick Pan's
factorization algorithm~\cite{Pan:alg} for the factorization step,  
which also currently offers the best worst case bit complexity.\footnote{\Kurt{In practice, one might consider a numerical root finder~\cite{Bini-Fiorentino} based on the Aberth-Ehrlich method instead. There is empirical evidence that such methods achieve comparable complexity bounds. We further remark that many solvers only provide approximations $\hz_1,\ldots,\hz_n$ of the roots without any guarantee on the error $\onenorm{p-p_n\cdot\prod_i(x-\hz_i)}$. In this case, we first have to estimate the latter error by an algorithm for approximate polynomial multiplication; e.g.~the method from~\cite{Schoenhage82} allows us to approximate the product $p_n\cdot\prod_i(x-\hz_i)$ to an absolute error of $2^{-b}$ using $\tilde{O}(n(n\Gamma_p+b))$ bit operations. Obviously, if $\hz_i\rightarrow z_i$ for all $i$, then $\onenorm{p-p_n\cdot\prod_i(x-\hz_i)}\rightarrow 0$, hence we can alternatively assume that our oracle provides arbitrary good approximations of the roots (without any additional estimate on the actual error).}}
If the candidates do not pass the verification step, we reapply
the factorization algorithm with a higher precision. Given a polynomial $p$ 
with $|z_i|\le 1$ for $1\le i\le k$, and a positive integer
$b$ denoting the desired precision, the factorization algorithm
computes $n$ root approximations $\hz_1,\ldots,\hz_n$.  The quality of
approximation and the bit complexity are as follows:

\begin{theorem}[Pan~\cite{Pan:alg}]\label{thm:pan}
Suppose that $|z_i|\le 1$ for $1\le i\le k$. For any positive integer $b\ge n\log
n$, complex numbers $\hz_1,\ldots,\hz_n$ can be computed such that they satisfy

\[  \onenorm{ p  - p_{n}\prod\nolimits_{i=1}^{n}(x-\hz_i) }
				\leq 2^{-b} \onenorm{ p }
			\]
			using $\tO(n)$ operations performed with the precision of $O(b)$
			bits (or $\tO(bn)$ bit-operations). The input to the algorithm is a relative $L$-approximation of $p$, where $L = O(b)$. We write
                        $\hp:=p_{n}\prod_{i=1}^{n}(x-\hz_i)$.
The algorithm returns the real and imaginary part of the $\hz_i$'s
 as dyadic fractions of the form $A \cdot 2^{-B}$ with $A \in \Z$, $B
\in \N$ and $B = O(b)$. All fractions have the same denominator.
		\end{theorem}
		
The parameter $b$ controls the quality of the resulting approximations. Note
that Pan's algorithm requires all roots of the input polynomial to lie within
the unit disk $\Delta(0,1)$.  Hence, in order to apply the above result to our input polynomial,
we first scale $p$ such that the roots come to lie in the unit disk.
That is, we compute a $\Gamma$ as in Theorem~\ref{thm:rootbound}, and then consider
the polynomial $f(x):=p(s\cdot x)=\sum_{i=0}^n f_i x^i$ with $s:=2^\Gamma$. Then, $f(x)$ has roots 
$\xi_{i}=z_i/s\in\Delta(0,1)$, and thus we can use Pan's Algorithm with $b':=n\Gamma+b$ to compute an
approximate factorization $\hat{f}(x):=\sum_{i=0}^n \hat{f}_i x^i := f_n\prod_{i=1}^{n}(x-\hat{\xi}_i)$ such that
$\onenorm{f-\hat{f}}<2^{-b'}\onenorm{f}$.
Let $\hz_i:=s\cdot
\hat{\xi}_i$ for all $i$ and $\hp(x):=p_n\cdot \prod_{i=1}^n (x-\hz_i) =
\hat{f}(x/s) = \sum_{i=0}^n \hat{f}_i/s^i x^i$, then 
\begin{align*}
\onenorm{\hp - p} &= \sum_{i=0}^n |f_i/s^i-\hat{f}_i/s^i| \le \sum_{i=0}^n
|f_i -\hat{f}_i| \le 2^{-b'}\sum_{i=0}^n |f_i|\le 2^{-b'}s^n\sum_{i=0}^n
|f_i/s^i|=2^{-b}\onenorm{p}.
\end{align*}
For the factorization of $f$, we need a relative $b'$-approximation of $f$, and
thus a relative $L$-approximation of $p$ with $L=O(n\Gamma+b)=\tilde{O}(n\Gamma_p+b)$. The
total cost is $\tilde{O}(n^2\Gamma_p+nb)$ bit operations. We summarize in:

\begin{corollary}\label{thm:pansresult}
For an arbitrary polynomial $p=p_n\cdot x^n+\cdots+p_0\in\C[x]$, with $1/4\le |p_n|\le 1$, and an integer $b\ge n\log
n$, complex numbers $\hz_1,\ldots,\hz_n$ can be computed such that
\[  \onenorm{ p  - p_{n}\prod\nolimits_{i=1}^{n}(x-\hz_i) }
				\leq 2^{-b} \onenorm{ p }
			\]
			using $\tO(n^2\Gamma_p+bn)$ bit-operations. We write
                        $\hp:=p_{n}\prod_{i=1}^{n}(x-\hz_i)$.
The algorithm returns the real and imaginary part of the $\hz_i$'s
 as dyadic fractions of the form $A \cdot 2^{-B}$ with $A \in \Z$, $B
\in \N$ and $B = O(b+n\Gamma_p)$. All fractions have the same denominator.
\end{corollary}
		
We now examine how far the approximations $\hz_1,\ldots,\hz_n$ can deviate from
the actual roots for a given value of $b$, \Kurt{i.e., a quantitative version of the fact that the roots of a polynomial depend continuously on the coefficients. Such estimates are well known, e.g.,~\cite[Theorem 2.7]{Schoenhage:GCD} and~\cite[Theorem 4.10c]{Henrici74}. For our complexity bounds, we also need the dependency on the multiplicities and the root separation and hence need to state our own bounds. Technically, there is nothing new here.}
Let $\Delta(z, r)$ be the disk with center $z$ and radius $r$ and let $\bd
\Delta(z,r)$ be its boundary. We further define $P_i := \prod_{j \not=i} (z_i -
z_j)^{m_j}$. Then, $p^{(m_i)}(z_i) = m_i ! p_n P_i$. 

\begin{lemma}\label{lem:value by dist}
If $r\le \sigma_i/n$, then
			\[ |p(x)| > \frac{r^{m_i}\cdot |p_n P_i |}{4} \]
for all $x$ on the boundary of $\Delta(z_i,r)$. 
\end{lemma}
\begin{proof}
We have
\begin{align*}
|p(x)|  &= \abs{p_n}\cdot|x-z_i|^{m_i}\cdot\prod_{j\neq i} |x-z_j|^{m_j} \ge\abs{p_n}\cdot|x-z_i|^{m_i}\cdot \prod_{j\neq i} 
|z_{i}-z_{j}|^{m_{j}}\cdot (1-|x-z_{i}|/|z_{i}-z_{j}|)^{m_{j}}\\
&\ge  r^{m_i} (1-1/n)^{n-m_i}\abs{p_n}\cdot\prod_{j\neq i}|z_i-z_j|^{m_j} > r^{m_i}\cdot|p_n P_i|/4.
\end{align*} \end{proof}
		
Based on the above Lemma, we can now use Rouch\'e's theorem\footnote{Rouch\'e's theorem states that if $f$
  and $g$ are holomorphic functions with $\abs{(f - g)(x)} < \abs{f(x)}$ for
  all points $x$ on the boundary of some disk $\Delta$, then $f$ and $g$ have
  the same number of zeros (counted with multiplicity) in
  $\Delta$.}
\ignore{\footnote{The reader may wonder why we are not using
  Gershgorin circles to cluster and verify root approximations. The problem is that very good approximations of a multiple root generate very large
Gershgorin circles which may then fail to isolate the approximations. Indeed, if two
approximations are equal, the corresponding disks have infinite radius.}} to show that,
for sufficiently large $b$, the disk  $\Delta(z_i, 2^{-b/(2m_i)})$ contains exactly $m_i$ root approximations.


\begin{lemma}\label{lem:approximation quality}
Let $\hp$ be such that $\onenorm{p-\hp} \le 2^{-b} \onenorm{p}$. If
\begin{align}
b &\ge  \max(8n,n \log ( n)) \text{, and $b$ is a power of two}\label{greaternlogn}\\
2^{-b/(2m_i)} & \le \frac{1}{2n^2},\label{greater2mi}\\
2^{-b/(2m_i)} & \le \frac{\sigma_i}{2n},\label{greater2milognoversigmai} \text{ and}\\
2^{-b/2} & \le \frac{|P_i|}{16(n+1)2^{\tau_p}M(z_{i})^{n}} \label{greaterr}
\end{align}
for all $i$,
the disk $\Delta(z_i, 2^{-b/(2m_i)})$ contains exactly $m_i$ root
approximations. For $i \neq j$, let $\hz_i$ and $\hz_{j}$ be arbitrary approximations in the disks $\Delta(z_i, 2^{-b/(2m_i)})$ and $\Delta(z_j, 2^{-b/(2m_j)})$, respectively. Then,
\[   \left(1 - \frac{1}{n}\right)\cdot \abs{z_i - z_j} \le \abs{\hz_i - \hz_j} \le \left(1 + \frac{1}{n}\right)\cdot  \abs{z_i - z_j}.\]
\end{lemma}
\begin{proof}
Let 
\[ \delta_i:=\left(16\cdot(n+1)\cdot 2^{-b}2^{\tau_p} |P_i|^{-1}M(z_{i})^{n}
                        \right)^{1/m_i}.\] 
It is easy to verify that $\delta_i \le 2^{-b/(2m_i)}\le \min(1,\sigma_i)/(2n)$.  The
first inequality follows from (\ref{greaterr}) and the second inequality
follows from (\ref{greater2mi}) and (\ref{greater2milognoversigmai}).
We will show that $\Delta(z_i,\delta_i)$ contains $m_i$ approximations. To
this end, is suffices to show that $|(p-\hp)(x)|<|p(x)|$ for all $x$ on the boundary of
$\Delta(z_i,\delta_i)$. Then, Rouch\'e's theorem guarantees that $\Delta(z_i,\delta_i)$ contains
the same number of roots of $p$ and $\hp$ counted with multiplicity. Since
$z_i$ is of multiplicity $m_i$ and $\delta_{i}<\sigma_{i}/n$, the disk contains exactly $m_i$ roots of $p$ counted with multiplicity.
We have (note that $\abs{x} \le (1+1/(2n^2))\cdot M(z_i)$ for $x \in \bd \Delta(z_i,\delta_i)$)
\begin{align*}
\abs{(p-\hp)(x)} &\le \onenorm{p-\hp}\cdot M(x)^n<2^{-b} \onenorm{p} M(x)^n\\
& \le  2^{-b} \onenorm{p}\cdot  (1+1/(2n^2))^n\cdot M(z_i)^n\\
&\le 4\cdot 2^{-b}\cdot 2^{\tau_p}|p_n|\cdot(n+1)\cdot M(z_i)^n\\
                          & \le \delta_i^{m_i}|p_n P_i |/4 < |p(x)|, 
\end{align*}
where the inequality in line three
follows from $\onenorm{p} \le (n+1) |p_n| 2^{\tau_p}$, the first one in line four
follows from the definition of $\delta_i$, and the last inequality follows from
Lemma~\ref{lem:value by dist}. 
It follows that $\Delta(z_i, 2^{-b/(2m_i)})$ contains exactly $m_{i}$ approximations. Furthermore, since $\delta_i \le \sigma_i/(2n)$ for all $i$, the disks $\Delta(z_i,\delta_i)$, $1 \le i
\le k$, are pairwise disjoint.

For the second claim, we observe that $\abs{\hz_\ell - z_\ell} \le 2^{-b/(2m_\ell)} \le \sigma_\ell/(2n)
\le \abs{z_i - z_j}/(2n)$ for $\ell =i,j$ and hence $\abs{\hz_i - z_i} + \abs{\hz_j - z_j} \le
\abs{z_i - z_j}/n$.  The claim now follows from the triangle inequality. 
\end{proof}

We have now established that the disks $\Delta(z_i,2^{-b/(2m_i)})$, $1 \le i \le k$, are pairwise
disjoint and that the $i$-th disk contains exactly $m_i$ root approximations
provided that $b$ satisfies (\ref{greaternlogn}) to
(\ref{greaterr}). We want to stress that the radii $2^{-b/(2m_i)}$, $1
  \le i \le k$, are vastly different. For example, assume $b = 40$. For a one-fold root ($m =
  1$), the radius is $2^{-20}$, for a double root ($m = 2$) the radius is
  $2^{-10}$, for a four-fold root ($m = 4$) the radius is $2^{-5}$, and for a
  twenty-fold root ($m = 20$), the radius is as large as $1/2$. 
Unfortunately, the conditions on $b$ are stated in terms of
the quantities $m_i$, $\sigma_i$ and $|P_i|$ which we do not know. Also, we do not know the center
$z_i$. In the remainder of the
section, we will show how to cluster root approximations and to certify
them. We will need the following more stringent properties for the clustering
and certification step. 
\begin{align}
2^{-b/(2m_i)} &<
\min\left(\left(\frac{\sigma_i}{4n}\right)^8,\frac{\sigma_i}{1024n^2}\right) \label{property 1}\\
2^{-b/8} &< \min\left(\frac{1}{16},\frac{\abs{P_i}}{(n+1) \cdot 2^{2n\Gamma_p+8n}}\right) \label{property 2}
\end{align}
Let $b_0$ be the smallest integer satisfying
(\ref{greaternlogn}) to (\ref{property 2}) for all $i$. Then,
$$b_{0}=O(n\log n+n\Gamma_{p}+
\max\nolimits_i (m_{i}\log M(\sigma_{i}^{-1})
+\log M(P_{i}^{-1})).$$

We next provide a high-level description of our algorithm to isolate the roots of $p$. The details of the clustering step and the certification step are then given in Sections~\ref{sec:clustering} and \ref{sec:certification}, respectively.

		
\subsubsection{Overview of the Algorithm} On input $p$ and the number $k$ of distinct roots,
the algorithm outputs isolating disks $\Delta_{i}=\Delta(\tz_i,R_i)$ for
the roots of $p$ as well as the corresponding multiplicities $m_i$. The radii satisfy
$R_i<\sigma_i/(64n)$.

The algorithm uses the factorization step with an increasing precision until the result can be certified. If either the clustering step or the certification step fails, we simply double the precision. There are a couple of technical safeguards to ensure that we do not waste time on iterations with an insufficiently large precision (Steps 2, 5, and 6); also recall that we need to scale our initial polynomial.\medskip

\ignore{\pengming{I think we should add a very high-level description of the algorithm first that stresses the key idea: That certification allows us to approximate a good value for $b$ without knowing $m_i$, $\sigma_i$, or $P_i$.}\michael{I think that this is already a relatively high level description. Maybe, we should say that steps 2, 5, and 6 are more of technical reason.}}

\begin{enumerate}
\item Compute the bound $2^{\Gamma}$ for the modulus of all roots of $p$, where
  $\Gamma$ fulfills Inequality~(\ref{inequ:rootbound}). According to
  Theorem~\ref{thm:rootbound}, this can be done with
  $\tilde{O}(n^{2}\Gamma_{p})$ bit operations.
\item Compute a 2-approximation $\lambda=2^{l_{\lambda}}$, with $l_{\lambda}\in\Z$, of $\onenorm{p}/|p_{n}|$. According to (\ref{apx:normp}), this computation needs $\tilde{O}(n\tau_{p})=\tO(n^2\Gamma_p)$ bit operations.\label{def:lambda}
\item Scale $p$, that is,
  $f(x):=p(s\cdot x)$, with $s:=2^\Gamma$, to ensure that the roots $\xi_i=z_i/S$, $i=1,\ldots,k$, of $f$ are contained in the unit disk. 
  Let $b$ be the smallest integer satisfying
  (\ref{greaternlogn})
\item Run Pan's algorithm on input $f$ with parameter $b':=b+n\Gamma$ to produce approximations
  $\hat{\xi}_1,\ldots,\hat{\xi}_n$ for the roots of $f$. Then, $\hat{z}_i:=s\cdot\hat{\xi}_i$ are approximations of the roots 
  of $p$, and $\onenorm{\hp-p}<2^{-b}\onenorm{p}$, where $\hp(x):=p_n\prod_{i=1}^n(x-\hz_i)$.\label{loop}
\item If there exists a $\hz_i$ with $\hz_i\ge 2^{\Gamma+1}$, return to Step~\ref{loop} with $b:=2b$.\label{algo:lower Gamma} 
\item If $\prod_{i=1}^n M(\hz_i)>8\lambda$, return to Step~\ref{loop} with $b:=2b$. \label{algo:lower lambda}
\item Partition $\hz_1,\ldots,\hz_n$ into $k$ clusters
  $C_1,\ldots,C_k$. Compute (well separated) enclosing disks $D_1,\ldots,D_k$
  for the clusters. For details, see Section~\ref{sec:clustering}. \label{cluster guarantee}
If the clustering fails to find $k$ clusters and corresponding disks, return to Step~\ref{loop} with $b:=2b$.
\item For each $i$, let $\Delta_i$ denote the disk with the same center as $D_i$ but with 
an $n$-times larger radius. We now verify
the existence of $|C_i|$ roots (counted with multiplicity) of $p$ in $\Delta_i$. For details of the verification, see Section~\ref{sec:certification}. If the verification fails, return to Step \ref{loop} with $b:=2b$. 
\item If the verification succeeds, output the disks $\Delta_i$ (in
  Step~\ref{cluster guarantee}, we guarantee that the disks $\Delta_{i}$ are pairwise disjoint) and report 
  the number $|C_i|$ of root approximations $\hz\in\{\hz_{1},\ldots,\hz_{n}\}$ contained in the disks as the corresponding
  multiplicities.\medskip
\end{enumerate}

Notice that Steps~\ref{algo:lower Gamma} and~\ref{algo:lower lambda} ensure that $\log M(\hz_i)=O(\Gamma_p+\log n)$ for all $i$, and that 
$\log\prod_{i=1}^n M(\hz_i)=O(\log(\onenorm{p}/|p_n|))=O(\log n+\tau_p)=\tilde{O}(n\Gamma_p)$.
The following Lemma guarantees that the algorithm passes these steps if $b\ge b_0$.

\begin{lemma}		
For any $b\ge b_0$, it holds that $|\hz_i|\ < 2^{\Gamma+1}$ for all $i$, and $\prod_{i=1}^n M(\hz_i)<8\lambda$.
\end{lemma}

\begin{proof}
In the proof of Lemma~\ref{lem:approximation quality}, we have already shown that $\abs{\hz_i} \le (1+1/(2n^2))\cdot M(z_i)$ for all $i$. Hence, it follows that $|\hz_i|\le (1+1/(2n^2))\cdot 2^{\Gamma_p}<2\cdot 2^{\Gamma_{p}}\le 2^{\Gamma_{p}+1}$, and 
\begin{align*}
\prod_{i=1}^n M(\hz_i)&\le 4\cdot \prod_{i=1}^k M(z_i)^{m_i}<\frac{4\Mea(p)}{|p_n|}
\le \frac{4\norm{p}_2}{|p_n|}\le \frac{4\onenorm{p}}{|p_n|}<8\lambda.
\end{align*}
\end{proof}

		
\subsubsection{Clustering}\label{sec:clustering}

After candidate approximations $\hz_1,\ldots,\hz_n$ are computed using a fixed
precision parameter $b$, we perform a partitioning of these approximations into
$k$ clusters $C_1,\ldots,C_k$, where $k$ is given as an input.  The clustering
is described in detail below. It works in phases. At the beginning of a phase,
it chooses an unclustered approximation and uses it as the seed for the cluster
formed in this phase. Ideally, each of the clusters 
corresponds to a distinct root of $p$. The clustering algorithm satisfies the following properties:
\begin{itemize} 
\item[(1)] For $b < b_0$, the algorithm may or may not succeed in finding $k$
 clusters.
\item[(2)] For $b \ge b_0$, the clustering always succeeds. 
\end{itemize}
Whenever the clustering succeeds, the cluster
$C_i$ with seed $\tz_i$ is contained in the disk $D_i
 :=\Delta(\tz_i,r_i)$, where $r_i \approx\min(\frac{1}{n^{2}},\frac{\tsigma_i}{256 n^2})$, and $\tsigma_i = \min_{j \not= i}
 \abs{\tz_i - \tz_j}$. Furthermore, for $b\ge b_{0}$, $D_{i}$ contains the root $z_{i}$ (under
 suitable numbering) and 
exactly $m_i$ many approximations.

Before we describe our clustering method, we discuss two evident approaches that do not work for any $b$ of size comparable to $b_{0}$ or smaller. A clustering with a fixed grid does not work as root
  approximations coming from roots with different multiplicities may move by vastly
  distinct amounts. As a consequence, we can only succeed if $b>(\max_{i}m_{i})\cdot \log(\min_{i}\sigma_{i})^{-1}$ which can be considerably larger than $b_{0}$, see Figure~\ref{fig:clustering}. A clustering based on
  Gershgorin disks does not work either because very good approximations of a multiple
  root lead to large disks which then fail to separate approximations of
  distinct roots. In particular, if approximations are identical, the
  corresponding Gershgorin disks have infinite radius.
		
For our clustering, we use the fact that the factorization algorithm provides approximations $\hz$ of the root $z_i$ with distance less than $2^{-b/(2m_i)}$ (for $b\ge b_0$).
Thus, we aim to determine clusters $C$ of maximal size such
that the pairwise distance between two elements in the same cluster is less
than $2\cdot 2^{-b/(2|C|)}$. We give details.\medskip 

\begin{enumerate}
\item Initialize ${\cal C}$ to the empty set (of clusters). 
\item \label{outer loop} Initialize $C$ to the set of all unclustered approximations and choose
  $\hz \in C$ arbitrarily. Let $a:=2^{\floor{\log n}+2}$ and $\delta:=2^{-b/4}$.
\item \label{inner loop} Update $C$ to the set of points $q\in C$ satisfying $|\hz-q| \le 2\sqrt[a/2]{\delta}$.
\item If $|C|\ge a/2$, add $C$ to ${\cal C}$. Otherwise, set $a:=a/2$ and continue with
  step~\ref{inner loop}. 
\item If there are still unclustered approximations, continue with
  step~\ref{outer loop}. 
\item \label{check 1} If the number of clusters in ${\cal C}$ is different
  from $k$, report failure, double $b$ and go back to the factorization step.\medskip
\end{enumerate}

\begin{figure}\centering
 \begin{tikzpicture}[scale=2]

 		\draw [fill=lightgray] (0.5,2.5) circle [radius=0.15];
		\draw [fill=lightgray] (0.4,2.15) circle [radius=0.15];
 		\draw [fill=lightgray] (0.7,1.2) circle [radius=0.4];
 		\draw [fill=lightgray] (2.3,2) circle [radius=0.9];
		
 		\draw [fill=black] (0.56,2.44) circle [radius=0.02];
		\draw [fill=black] (0.50,2.2) circle [radius=0.02];
 		\draw [fill=black] (0.8,1.2) circle [radius=0.02];
 		\draw [fill=black] (0.8,1.4) circle [radius=0.02];
 		\draw [fill=black] (1.5,2.1) circle [radius=0.02];
 		\draw [fill=black] (2.5,2.2) circle [radius=0.02];
 		\draw [fill=black] (2.4,1.7) circle [radius=0.02];
 		\draw [fill=black] (2.8,2) circle [radius=0.02];
		
 		\node at (0.7,1.2) {$\times$};
 		\node at (0.5,2.5) {$\times$};
                  \node at (0.4,2.15) {$\times$};
 		\node at (2.3,2) {$\times$};

 		\node at (0.5,2.8) {$r=2^{-b/2}$};
		\node at (0.6,1.9) {$r=2^{-b/2}$};
 		\node at (0.9,0.7) {$r=2^{-b/4}$};
 		\node at (3,1.1) {$r=2^{-b/8}$};
		
 	\end{tikzpicture}
	\caption{\label{fig:clustering}\textmd{Example of a polynomial with four distinct roots with multiplicities 1, 1, 2, and 4. Crosses are roots of the polynomial, dots represent the approximations. The disk around a root shows the potential locations of its approximations. Note that the straight-forward approach to cluster with a fixed distance threshold fails for all $b$ with $b<(\max_{i}m_{i})\cdot\log (\min_{i}\sigma_{i})^{-1}$: For each such $b$, one can not choose any threshold that allows detecting the simple roots without splitting the four-fold root.}}
\end{figure}

Note that, for $b \ge b_0$, the disks $\Delta(z_i, 2^{-b/(2m_i)})$ are
disjoint. Let $Z_i$ denote the set of root
approximations in $\Delta(z_i, 2^{-b/(2m_i)})$.  Then, $\abs{Z_i} = m_i$ according to Lemma~\ref{lem:approximation quality}.
We show that, for $b \ge b_0$, the clustering algorithm terminates with $C=Z_i$ if called
with an approximation $\hz\in Z_i$.
		
\begin{lemma}\label{lem:distbetweenapx}
Assume $b \ge b_0$, $\hz_i \in Z_i$, $\hz_j \in Z_j$, and $i \not= j$. Then,
\begin{align*}
\abs{\hz_i-\hz_j} &\ge 2 \cdot ( 2^{-b/(16m_i)} + 2^{-b/(16m_j)}).
\end{align*}
\end{lemma}
\begin{proof}
Since $b \ge b_0$, we have $2^{-b/(2m_\ell)} \le \sigma_\ell$ for $\ell =
i,j$ by (\ref{greater2milognoversigmai}) and $2^{-b/(16m_\ell)} =
(2^{-b/(2m_\ell)})^{1/8} \le \sigma_\ell/(4n) \le \sigma_\ell/8$ by (\ref{property
  1}). Thus,
\begin{align*}
\abs{\hz_i-\hz_j}	\ge \max(\sigma_i,\sigma_j) - 2^{-b/(2m_i)} - 2^{-b/(2m_j)} 
		\ge \frac{\sigma_i}{2} + \frac{\sigma_j}{2} - \frac{\sigma_i}{4} - \frac{\sigma_j}{4}\ge 2\cdot ( 2^{-b/(16m_i)} + 2^{-b/(16m_j)}).
\end{align*} 
\end{proof}
		
\begin{lemma}\label{lem:clustering correctness} If $b \ge b_0$, the clustering
  algorithm computes the correct clustering, that is, it produces clusters $C_1$
  to $C_k$ such that $C_i = Z_i$ for all $i$ (under suitable numbering). Let
  $\tz_i$ be the seed of $C_i$ and let $\tsigma_i = \min_{j \not= i} \abs{\tz_i
    - \tz_j}$. Then, $(1 - 1/n)\sigma_i \le \tsigma_i \le (1 + 1/n) \sigma_i$ and $C_i$ as well as the root $z_{i}$ is contained in
  $\Delta(\tz_i,\min(\frac{1}{n^{2}},\frac{\tsigma_i}{256 n^2}))$. 

\end{lemma}
\begin{proof} Assume that the algorithm has already produced $Z_1$ to
  $Z_{i-1}$ and is now run with a seed $\hz \in Z_i$. We prove that it
  terminates with $C = Z_i$. Let $\ell$ be a power of two such that $\ell \le
  m_i < 2\ell$. The proof that the algorithm terminates with $C = Z_i$ consists
  of two parts. We first assume that steps 2 and 3 are executed for $a = 2
  \ell$. We show that the algorithm will then terminate with $C = Z_i$. In the
  second part of the proof, we show that the algorithm does not terminate as
  long as $a > 2\ell$. 
			
Assume the algorithm reaches steps 2 and 3 with $a/2=\ell$, i.e. $a/2\le m_i <
a$. For any approximation $q\in Z_i$, we have $|\hz-q|\le 2\cdot 2^{-b/(2m_i)}
= 2\sqrt[m_i/2]{\delta} \le 2\sqrt[a/2]{\delta}$. Thus, $Z_i\subseteq
C$.  Conversely, consider any approximation $q\notin Z_i$. Then, 
$\abs{\hz-q} \ge 2\cdot 2^{-b/(16m_i)} > 2\sqrt[4m_i]{\delta}\ge
2\sqrt[2a]{\delta}$, and thus no such approximation is contained in $C$. This shows that 
$C = Z_i$. Since $\abs{C} \ge a/2$, the algorithm terminates and returns
$Z_i$.

It is left to argue that the algorithm does not terminate before $a/2= \ell$. Since
$\ell$ and $a$ are powers of two, assume we terminate with $a/2\ge 2\ell$, and
let $C$ be the cluster returned. Then, $m_i < a/2 \le \abs{C} < a$ and $Z_i$ is
a proper subset of $C$. Consider any approximation $q\in C \setminus Z_i$, say
$q \in Z_j$ with $j \not= i$. 
Since $q\notin Z_i$,  we have $|q-\hz|\ge 2\cdot(
2^{-b/(16m_i)} + 2^{-b/(16m_j)}) > 2 \cdot
2^{-b/(16m_i)}> 2\sqrt[4m_i]{\delta}$. And since $q\in C$, we have
$|q-\hz|\le 2\sqrt[a/2]{\delta}$. Thus, $4m_i\le a/2$ and, hence, there are 
at least $3a/8$ many approximations in $C \setminus Z_i$. Furthermore, 
$\abs{z_i - z_j} \le \abs{z_i - \hz} + \abs{\hz - q} + \abs{q - z_j} \le
2^{-b/(2m_i)} + 2 \sqrt[a/2]{\delta} + 2^{-b/(2m_j)} \le 2^{-b/(16m_i)} + 2
\sqrt[a/2]{\delta} + 2^{-b/(16m_j)} \le 3
\sqrt[a/2]{\delta}$. Consequently,
there are at least $3a/8$ roots $z_j\neq z_i$ counted with multiplicity within distance
$3 \sqrt[a/2]{\delta}$ to $z_i$. This observation allows us to upper
bound the value of $|P_i|$, namely
\begin{align*}
|P_i| = \prod_{j\neq i}|z_i-z_j|^{m_j} \le (3\sqrt[a/2]{\delta})^{3a/8} 2^{(n - m_i - 3a/8)\Gamma_p}< 3^n \delta^{3/4} 2^{n\Gamma_p}\le 3^n 2^{-3b/16}2^{n\Gamma_p}<3^{n} 2^{-b/8}\cdot 2^{n\Gamma_p},
\end{align*}
a contradiction to (\ref{property 2}). 

We now come to the claims about $\tsigma_i$ and the disks defined in terms of
it. The relation between $\sigma_i$ and $\tsigma_i$ follows from the second
part of Lemma~\ref{lem:approximation quality}. All points in $C_i = Z_i$ have distance at most $2\cdot 2^{-b/(2m_i)}$
from $\tz_i$. Also, by (\ref{greater2mi}) and (\ref{property 1}),
\[ 2\cdot 2^{-b/(2m_i)} < \min(1/n^2,\sigma_i/(512 n^2)) \le
\min(1/n^2,\tsigma_i/(256n^2))\]
Hence, $C_i$ as well as $z_{i}$ is contained in $\Delta(\tz_i,\min(1/n^2,\tsigma_i/(256n^2)))$.  
\end{proof}
		
\begin{lemma}\label{lem:clustering complexity}
	For a fixed precision $b$, computing a complete clustering needs $\tO(nb+n^2\Gamma_p)$ bit operations.
\end{lemma}
\begin{proof}
For each approximation, we examine the number of distance computations we need
to perform. Recall that $b$ (property (\ref{greaternlogn})) and $a$ are powers
of two, $a \le 4n$ by definition, and $b \ge 8n\ge 2a$ by property
(\ref{greaternlogn}). Then, $\sqrt[a/2]{\delta} = 2^{-b/(2a)} \in 2^{-\N}$. Thus, the number $\sqrt[a/2]{\delta}$ has a very
simple format in binary notation. There is a single one, and this one is
$b/(2a)$ positions after the binary point. In addition, all approximations $\hz$ have absolute value less than $2\cdot 2^\Gamma$ due to Step~\ref{algo:lower Gamma} in the overall algorithm.
Thus, each evaluation of the form $|\hz - q|\le
2\sqrt[a/2]{\delta}$ can be done with $$O(\Gamma+\log \delta^{-2/a})=O((b/a)+\Gamma)=O((b/a)+\Gamma_p+\log n)$$
bit operations.
	
For a fixed seed $\hz$, in the $i$-th iteration of step 2, we have at most $a\le n/2^{i-2}$ many unclustered approximations left in $C$, since otherwise we would have terminated in an earlier iteration. 
	Hence, we perform at most $a$ evaluations of the form $|\hz - q|\le 2\sqrt[a/2]{\delta}$, resulting in an overall number of bit operations of $a\cdot O((b/a)+\Gamma)=O(b+a\Gamma)$ for a fixed iteration. As we halve $a$ in each iteration, we have at most $\log n+2$ iterations for a fixed $\hz$, leading to a bit complexity of $O(b\log n+n\Gamma)=\tO(b+n\Gamma)=\tO(b+n\Gamma_p)$.
	
	In total, performing a complete clustering has a bit complexity of at
        most $\tO(nb+n^2\Gamma_p)$.
\end{proof}
		
When the clustering succeeds, we have $k$ clusters $C_1$ to $C_k$ and
corresponding seeds $\tilde{z}_{1},\ldots,\tilde{z}_k\subseteq \{\hz_{1},\ldots,\hz_{n}\}$. For $i=1,\ldots,k$, we define $D_i := \Delta(\tz_i,r_i)$, where $\tz_i$ is the seed for the cluster $C_{i}$ and 
\begin{align}\label{def:ri}
r_i:=\min(2^{-\ceil{2\log n}},  2^{\ceil{\log \tilde{\sigma}_{i}/(256n^{2})}})\ge\min\left(\frac{1}{2n^2},\frac{\tilde{\sigma_{i}}}{256n^{2}}\right).
\end{align}
In particular, $r_{i}$ is a $2$-approximation of $\min(1/n^2,\tilde{\sigma}_{i}/(256n^{2}))$. Notice that the cost for computing the separations $\tilde{\sigma}_{i}$ is bounded by $\tilde{O}(nb+n^{2}\Gamma_p)$ bit operations since we can compute the nearest neighbor graph of the points $\tz_{i}$ (and thus the values $\tilde{\sigma}_{i}$) in $O(n\log n)$ steps~\cite{EPY-NNG:1997} with a precision of $O(b+n\Gamma)$.

Now, suppose that $b\ge b_{0}$, Then, according to Lemma~\ref{lem:clustering correctness}, the cluster $C_i$ is contained in the disk $D_i$. Furthermore, $D_i$ 
contains exactly one root $z_i$ of $p$ (under suitable numbering of the roots), and it holds that 
$m_i=\operatorname{mult}(z_{i},p)=|C_{i}|$ and $\min(1/(2n^2),\sigma_{i}/(512n^{2}))\le r_i\le \min(1/n^2,\sigma_{i}/(64n^{2}))$. If the clustering succeeds for a $b < b_0$, we have no
   guarantees (actually, the termination condition in step 4 gives some
   guarantee, however, we have chosen not to exploit it). Hence, before we proceed, we verify that each disk $D_i$ actually contains the cluster $C_{i}$. If this is not the case, then we report a failure, return to the factorization step with $b=2b$, and compute a new corresponding clustering.
   
   In the next and final step, we aim to show that each of the enlarged disks
  $\Delta_{i}:=\Delta(\tilde{z}_{i},R_{i}):=\Delta(\tilde{z}_{i},nr_{i})$,
  $i=1,\ldots,k$, contains exactly one root $z_i$ of $p$, and that the number
  of elements in $C_i\subseteq \Delta_{i}$ equals the multiplicity of
  $z_i$. Notice that, from the definition of $r_{i}$ and $\Delta_{i}$, it
  obvious that the disks $\Delta_{i}$ are pairwise disjoint and that
  $C_{i}\subseteq D_{i}\subseteq \Delta_{i}$.

		
\subsubsection{Certification}\label{sec:certification}

\ignore{The certification step receives clusters $C_1$ to $C_k$, their seeds
$\tz_1$ to $\tz_k$, and corresponding enclosing disks
$D_{1}=\Delta(\tz_{1},r_{1})$ to $D_k=\Delta(\tz_{1},r_k)$ with
$\min(\frac{1}{2n^2},\frac{\tilde{\sigma}_{i}}{256n^{2}})\le r_i\le
\min(\frac{1}{n^2},\frac{\tilde{\sigma}_{i}}{128n^{2}})$. Furthermore, it holds
that $C_{i}\subseteq  D_{i}\subseteq  \Delta_i=\Delta(\tz_{i},nr_{i})$ and the disks $\Delta_{i}$ are pairwise disjoint.} In order to show that $\Delta_{i}$ contains exactly one root of $p$ with multiplicity $|C_{i}|$, we show that each $\Delta_{i}$ contains the same number of roots of $p$ and $\hp$ counted with multiplicity. For the latter, we compute a lower bound for $|\hp(z)|$ on the boundary $\bd\Delta_{i}$ of $\Delta_{i}$, and check whether this bound is larger than $|(\hp-p)(z)|$ for all points $z\in\bd\Delta_{i}$. If this is the case, then we are done according to Rouch\'{e}'s theorem. Otherwise, we start over the factorization algorithm with $b=2b$. We now come to the details:\medskip

\begin{enumerate}
\item Let $\lambda=2^{l_{\lambda}}$ be the $2$-approximation of $\onenorm{p}/|p_n|$ as defined in step~\ref{def:lambda} of the overall algorithm. 

\item \label{step 2} For $i=1,\ldots,k$, let $z_i^*:=\tz_i+n\cdot r_i\in\Delta_i$. Note that $|z_i^*|\le(1+1/n)\cdot M(\tz_i)$ since $nr_i\le 1/n$.
\item We try to establish the inequality 
\begin{equation}\label{eq:cert}
|\hp(z_i^*)/p_{n}| > E_{i}:=64\cdot 2^{-b}\lambda M(\tz_i)^n
\end{equation} for all $i$. We will see in the proof of
  Lemma~\ref{lem:value on D+} that this implies
  that each disk $\Delta_{i}$ contains exactly one root
  $z_i$ of $p$ and that its multiplicity equals the number $|C_{i}|$ of approximations within $\Delta_{i}$. In order to
  establish the inequality, we consider $\rho=1,2,4,8,\ldots$ and compute $\abs{\hp(z_i^*)/p_n}$ to an absolute error less than $2^{-\rho}$. If, for all $\rho\le b$, we fail to show that $\abs{p(z_i^*)/p_n}>E_{i}$, we report a failure and go back to the factorization algorithm with $b=2b$. Otherwise, let $\rho_{i}$ be the smallest $\rho$ for which we are successful. 
  
  \item If, at any stage of the algorithm, $\sum_{i}\rho_{i}>b$, we also report a failure and go back to the factorization algorithm with $b=2b$. 
Lemma~\ref{lem:costevaluations} then shows that, for fixed $b$, the number of bit operations that are used for all evaluations is bounded by $\tilde{O}(nb+n^2\tau_p+n^{3})$. 
\item If we can verify that $|\hp(\tz_i+nr_i)/p_{n}|> E_{i}$ for all $i$, we return the disks $\Delta_{i}$ and the multiplicities $m_{i}=|C_{i}|$.\medskip
	\end{enumerate}

\begin{lemma}\label{lem:costevaluations}
For any $i$, we can compute $|\hp(z_i^*)/p_n|$ to an absolute error less than $2^{-\rho}$ with a number of bit operations less than $$\tilde{O}(n(n+\rho+n\log M(\tz_i)+\tau_p)).$$ For a fixed $b$, the total cost for all evaluations in the above certification step is upper bounded by $\tilde{O}(nb+n^2\tau_p+n^{3})$.
\end{lemma}

\begin{proof}
Consider an arbitrary subset $S\subseteq \{\hz_{1},\ldots,\hz_{n}\}$. We first derive an upper bound for $\prod_{\hz\in S}|z_i^{*}-\hz|$. For that, consider the 
polynomial $\hp_{S}(x):=\prod_{\hz\in S}(x-\hz)$. The $i$-th coefficient of $\hp_{S}$ is bounded by 
$\binom{|S|}{i}\cdot \prod_{\hz\in S}M(\hz)\le 2^{n}\prod_{i=1}^{n}M(\hz_{i})\le 8\lambda\cdot 2^{n}$ due to step~\ref{algo:lower lambda} in the overall algorithm.
It follows that 
\begin{align*}
\prod_{\hz\in S}|z_i^{*}-\hz|&=|\hp_{S}(z_{i}^{*})|\le (n+1)M(z_{i}^{*})^{n}\cdot 8\lambda\cdot 2^{n}< 64(n+1)^{2}\cdot 2^{n}2^{\tau_{p}}M(\tz_{i})^{n}
\end{align*}

In order to evaluate $|\hp(z_i^*)/p_n|=\prod_{j=1}^{n}|z_{i}^{*}-\hz_{j}|$, we use approximate interval evaluation with an absolute precision $K=1,2,4,8,\ldots$. More precisely, we compute the distance of $z_{i}^{*}$ to each of the points $\hz_{j}$, $j=1,\ldots,n$, up to an absolute error of $2^{-K}$, and then take the product over all distances using a fixed point precision of $K$ bits after the binary point.\footnote{In fact, we compute an interval $I_{j}$ of size less than $2^{-K}$ such that $|z_{i}^{*}-\hz_{j}|\in I_{j}$, and then consider the product $\prod_{j}I_{j}$.}
We stop when the resulting interval has size less than $2^{-\rho}$. The above consideration shows that all intermediate results have at most $O(n+\tau_{p}+n\log M(\tz_{i})$ bits before the binary point. Thus, we eventually succeed for an $K=O(\rho+\tau_p+n+n\log M(\tz_i))$. Since we have to perform $n$ subtractions and $n$ multiplications, the cost is bounded by $\tilde{O}(nK)$ bit operations for each $K$. Hence, the bound for the evaluation of $|\hp(z_i^*)/p_n|$ follows.

We now come to the second claim. Since we double $\rho$ in each iteration and consider at most $\log b$ iterations, the cost for the evaluation of $|\hp(z_i^*)/p_n|$ are bounded by $\tilde{O}(n(n+\rho_{i}+n\log M(\tz_i)+\tau_p))$. Since we ensure that $\sum_{i}\rho_{i}\le b$, it follows that the total cost is bounded
by $\tilde{O}(nb+n^2\tau_p+n^{3}+n^2\log(\prod_{i=1}^k M(\tz_i)))$. The last summand is smaller than $n^2\cdot 8\lambda$ 
according to step~\ref{algo:lower lambda}, and $\lambda<2\onenorm{p}/|p_n|<2(n+1)2^{\tau_p}$. This shows the
claim.
\end{proof}	
	
	We now prove correctness of the certification algorithm. In particular, we show that Inequality~(\ref{eq:cert}) implies that the disk $\Delta_{i}$ contains the same number of roots of the polynomials $\hp$ and~$p$.
	\begin{lemma}\label{lem:value on D+}
\begin{enumerate}
\item For all points $x\in \bd \Delta_i$, it holds that
		\[ |\hp(x)|\ge \frac{1}{\KurtTwo{4}{8}}|\hp(z_i^*)|\ . \]
\item If Inequality (\ref{eq:cert}) holds for all $i$, then $\Delta_i$ isolates a root of
    $z_i$ of $p$ of multiplicity $m_{i}=\operatorname{mult}(z_{i},p)=|C_{i}|$.
\item If $b \ge b_0$, then 
\[    \frac{|\hp(z_i^*)|}{|p_{n}|} >
\left(\frac{\min(256,\sigma_i)}{1024n}\right)^{m_i}\cdot \frac{|P_i|}{8} \ge 64\cdot 2^{-b_0}\lambda M(\tz_i)^n\]
\end{enumerate}
	\end{lemma}

	\begin{proof}
		\KurtTwo{First, we show that, for any two points $x,y\in \bd \Delta_i$, their distances to any approximation $\hz\in D_j$ differ only by a factor in between $[1-1/n,1+1/n]$:
		\[ (1-1/n)|y-\hz| \le |x-\hz| \le (1+1/n)|y-\hz| \ .\]}{}
		For a fixed $\hz$, let $x$ be the farthest point on \KurtTwo{$\bd \Delta$}{$\bd \Delta_i$} from $\hz$, and let $y$ be the nearest. For $i\neq j$, we have
		\begin{align*}
			|x-\hz| &\le |x-\tz_i| + |\tz_i-\tz_j| + |\tz_j-\hz| \le (1+1/n)|\tz_i-\tz_j|\ , \text{ and }\\
			|y-\hz| &\ge |\tz_i-\tz_j| - |y-\tz_i| - |\tz_j-\hz| \ge (1-1/n)|\tz_i-\tz_j|\ .
		\end{align*}
		Similarly, for $i=j$:
		\begin{align*}
			|x-\hz| &\le |x-\tz_i| + |\tz_i-\hz| \le (1+1/n)nr_i\ , \text{ and }\\
			|y-\hz| &\ge |y-\tz_i| - |\tz_i-\hz| \ge (1-1/n)nr_i\ .
		\end{align*}
		Consequently, for any $x,y\in \bd \Delta_i$, it holds that
		\[ |\hp(x)| = |p_n|\cdot\prod_{\ell=1}^{n} |x-\hz_\ell|\ge \left(\frac{1-1/n}{1+1/n}\right)^n |p_n|\cdot\prod_{\ell=1}^{n} |y-\hz_\ell| \ge
                \frac{1}{8}|\hp(y)|\ .\]
This shows the first claim. 

We turn to the second claim. Since $nr_i<1/n$, we have $|x|<(1+1/n)M(\tz_i)$ for all $x\in\bd\Delta_i$. Now, if 
$|\hp(z_i^*)/p_{n}| > 64\cdot 2^{-b}\lambda M(\tz_i)^n$, then 
\begin{align*}
|\hp(x)|>\frac{|\hp(z_i^*)|}{8}>2|p_n|\lambda 2^{-b} M(x)^n>\onenorm{p}2^{-b}M(x)^n \ge \onenorm{\hp - p} M(x)^n \ge |\hp(x)-p(x)|.
\end{align*}
Hence, according to Rouch\'e's theorem. $\Delta_{i}$ contains the same number (namely, $|C_{i}|$) of roots of $p$ and $\hp$. If this holds for all disks $\Delta_{i}$, then each of the disks must contain exactly one root since $p$ has $k$ distinct roots. In addition, the multiplicity of each root equals the number $|C_{i}|$ of approximations within $\Delta_{i}$.

It remains to show the third claim. Since $b \ge b_0$, it follows that $\min(1/(2n^2),\sigma_{i}/(512n^{2}))\le r_i\le \min(1/n^2,\sigma_{i}/(64n^{2}))$ and $|\tz_{i}-z_{i}|<r_{i}$; cf. the remark following the definition of $r_{i}$ in (\ref{def:ri}). Thus,
\begin{align*}
|\hp(z_i^*)| &\ge |p(z_i^*)|  - 2^{-b}\onenorm{p} \cdot M(z_i^*)^n\\
&= |p(z_i + (\tz_i  - z_i +nr_i))|  - 2^{-b}\onenorm{p}\cdot M(z_i^*)^n \\
&\ge ((n-1)r_i)^{m_i}|p_nP_i|/4 - 4\cdot 2^{-b}\onenorm{p} M(z_i)^n \\
&\ge \left(\frac{(n-1)\min(256,\sigma_i)}{512n^2}\right)^{m_i}\cdot \frac{|p_nP_i|}{4} -
4\cdot 2^{-b}\onenorm{p} M(z_i)^n\\
&\ge \left(\frac{\min(256,\sigma_i)}{1024n}\right)^{m_i}\cdot \frac{|p_nP_i|}{4} -
4\cdot 2^{-b}\onenorm{p}M(z_i)^n,
\end{align*}
where the first inequality is due to $|(p-\hp)(x)|<2^{-b}\onenorm{p}\cdot M(x)^n$, the second inequality follows
from $\abs{\tz_i  - z_i +nr_i} \le (n+1)r_i \le \sigma_i/n$, Lemma \ref{lem:value by dist} and $M(z_i^*)<(1+1/n)\cdot M(z_i)$, and the third inequality follows from $r_i \ge
\min(\frac{1}{2n^{2}},\frac{\sigma_i}{512 n^2})$. In addition, we have
\begin{equation}
\label{eq:useful} 2^{-b} \onenorm{p} M(z_i)^n \le
\left(\frac{\min(256,\sigma_i)}{1024n}\right)^{m_i}\cdot
\frac{|p_nP_i|}{4096},\end{equation}
since
\begin{align*}
2^{-b} \onenorm{p} M(z_i)^n
&\le 2^{-b/8} \cdot 2^{-b/2} \cdot
2^{\tau_p} \abs{p_n}\cdot (n+1) \cdot  M(z_i)^n\\
&\le \frac{\abs{P_i}}{(n+1)
  2^{2n\Gamma_p+8n}}\left(\frac{\min(256,\sigma_i)}{1024 n}\right)^{m_i}
2^{\tau_p} \abs{p_n} (n+1) M(z_i)^n\\
&\le \left(\frac{\min(256,\sigma_i)}{1024n}\right)^{m_i}\cdot
\frac{|p_nP_i|}{2^{7n - 1}} \le \left(\frac{\min(256,\sigma_i)}{1024n}\right)^{m_i}\cdot \frac{|p_nP_i|}{4096},
\end{align*}
where the second inequality follows from (\ref{property 2}), (\ref{property
  1}), and (\ref{greater2mi})~\footnote{Observe $2^{-b/(2m_i)} \le
  \min(\frac{1}{2n^2},\frac{\sigma_i}{1024n^2}) \le \frac{\min(256,\sigma_i)}{1024 n}$.}, and the
third inequality follows from $\tau_p \le n \Gamma_p + n + 1$ (Lemma~\ref{lem:tauTau})
and $M(z_i)^n \le 2^{n\Gamma_p}$. Finally,
\begin{align*}
\frac{|\hp(z_i^*)|}{|p_n|}> \left(\frac{\min(256,\sigma_i)}{1024n}\right)^{m_i}\cdot \frac{|P_i|}{8}\ge
512 \cdot 2^{-b}\frac{\onenorm{p}}{|p_n|}M(z_i)^n\ge 64\cdot 2^{-b}\lambda M(\tz_i)^n,
\end{align*}
where the first and the second inequality follow from (\ref{eq:useful}) and the
third inequality holds since $\lambda$ is a 2-approximation of
$\onenorm{p}/\abs{p_n}$ and $\abs{z_i}^n \le (1 + 1/n)^n \abs{\tz_i}^n \le 4\abs{\tz_i}^n$. Since the values $\sigma_i$, $m_i$, and $P_i$ do not depend on the choice of $b$, and the above inequality holds for any $b\ge b_0$, it follows that $512 \cdot 2^{-b}\frac{\onenorm{p}}{|p_n|}M(z_i)^n\ge 64\cdot 2^{-b_0}\lambda M(\tz_i)^n$.
\end{proof}

\begin{lemma}\label{def:certcomplexity}
There exists a $b^{*}$ upper bounded by
\[
O\left(n\log n+n\Gamma_{p}+\sum\nolimits_{i=1}^{k}\left(\log M(P_{i}^{-1})+m_{i}\log M(\sigma_{i}^{-1})\right)\right)
\]
such that the certification step succeeds for any $b>b^{*}$. The total cost in the certification algorithm (i.e.~for all iterations until we eventually succeed) is 
bounded by 
$$\tilde{O}\left(n^{3}+n^{2}\tau_p+n\cdot\sum\nolimits_{i=1}^k \left(\log M(P_{i}^{-1})+ m_{i}\log M(\sigma_{i}^{-1})\right)\right)$$
bit operations.
\end{lemma}		

\begin{proof}
Let $b\ge b_0$. Then, due to Lemma~\ref{lem:value on D+},
$$|\hp(z_i^*)/p_{n}|>\left(\frac{\min(256,\sigma_i)}{1024n}\right)^{m_i}\cdot
\frac{|P_i|}{8}\ge 64\cdot 2^{-b_{0}}\lambda M(\tz_i)^n\ge 128\cdot 2^{-b}\lambda M(\tz_i)^n$$ Thus, in order to verify
inequality (\ref{eq:cert}), it suffices to evaluate $|\hp(z_i^*)/p_{n}|$ to an
error of less than $|\hp(z_i^*)/2p_{n}|$. It follows that we
succeed for some $\rho_{i}$ with $$\rho_{i}=O(m_{i}\log n
+m_{i}\max(1,\log\sigma_{i}^{-1})+\log \max(1,|P_{i}|^{-1})).$$  
In Step 3 of the certification algorithm, we require that the sum over all $\rho_{i}$ does not exceed $b$. Hence, we eventually succeed in verifying the inequality (\ref{eq:cert}) for all $i$ if $b$ is larger than some $b^{*}$ with
\begin{align*}
b^{*}&=O(b_{0}+\sum\nolimits_{i} m_{i}\log n+\sum\nolimits_{i}(\log M(P_{i}^{-1})+m_{i}\log M(\sigma_{i}^{-1})))\\
&=O(n\log n+n\Gamma_{p}+\sum\nolimits_{i}(\log M(P_{i}^{-1})+m_{i}\log M(\sigma_{i}^{-1}))).
\end{align*}
For the bound for the overall cost, we remark that, for each $b$, the certification algorithm needs $\tO(n^{3}+nb+n^2\tau_p)$
bit operations due to Lemma~\ref{lem:costevaluations}. Thus, the above bound follows from the fact that that we double $b$ in each step and that the certification algorithm succeeds under guarantee for all $b>b^*$.
\end{proof}

\ignore{
\michael{I did the rewriting until this point. The complexity is now straight
  forward. The final bound is the same as in the Lemma above. Notice that the
  term $n^2\tau_p$ can also be replaced by $n^3+n^2\Mea/|p_n|$ because
  $2^{\tau_p}<2^n\Mea(p)/|p_n|$. This is particularly nice because, then, the
  bound only depends on $n$ and the geometry of the roots. You will notice that
  we could also replace $n^2\tau_p$ by $n^3\Gamma_p$, however, this weakens the
  bound for integer polynomials. The reason is that $\tau_p$ is strongly
  overestimated by $n\Gamma_p$ if $p$ has integer coefficients. I am convinced
  that the $n^{3}$ term should not be there, even though it does not weaken the
  total degree. However, so far, I have no idea how to avoid it in the proof of
  Lemma~\ref{lem:costevaluations}.}}



\subsection{Complexity of Root Isolation}
	We now turn to the complexity analysis of the root isolation algorithm. In the first step, we provide a bound for general polynomials $p$ with real coefficients. In the second step, we give a simplified bound for the special case, where $p$ has integer coefficients. We also give bounds for the number of bit operations that is needed to refine the isolating disks to a size less than $2^{-\kappa}$, with $\kappa$ an arbitrary positive integer.
	
	\begin{theorem}\label{thm:isolation complexity}
		Let $p(x)\in\C[x]$ be a polynomial as defined in Section~\ref{sec:setting}. We assume that the number $k$ of distinct roots of $p$ is given. 
		Then, for all $i=1,\ldots,k$, the algorithm from
                Section~\ref{sec:algorithm} returns an isolating disk
                $\Delta(\tz_{i},R_{i})$ for the root $z_{i}$ together with the
                corresponding multiplicity $m_{i}$, and $R_{i} < \frac{\sigma_{i}}{64n}$. 
		
		For that, it uses a number of bit operations bounded by
		\begin{align}\label{complexityresult}
		\tilde{O}\left(n^{3}+n^{2}\tau_p+n\cdot\sum\nolimits_{i=1}^k \left(\log M(P_{i}^{-1})+ m_{i}\log M(\sigma_{i}^{-1})\right)\right)
		\end{align}
	The algorithm needs an absolute $L$-approximation of $p$, with $L$ bounded by
	\begin{align}\label{apxquality}
	\tO\left(n\Gamma_p+\sum\nolimits_{i=1}^k \left(\log M(P_{i}^{-1})+ m_{i}\log M(\sigma_{i}^{-1})\right)\right).
	\end{align} 
				
	\end{theorem}
	\begin{proof}
		For a fixed $b$, let us consider the cost for each of the steps in the algorithm:		\begin{itemize}
		\item Steps 1-3, 5 and 6 do not use more than $\tilde{O}(n^{2}\Gamma_{p}+nb)$ bit operations,
		\item Step 4 and 7 do not use more than $\tilde{O}(n^{2}\Gamma_{p}+nb)$ bit operations (Corollary~\ref{thm:pansresult} and Lemma~\ref{lem:clustering complexity}), and
		\item Step 8 and 9 use a number of bit operations bounded by (\ref{complexityresult})
		(Lemma~\ref{def:certcomplexity}).
		\end{itemize}

In addition, for a fixed $b$, the oracle must provide an absolute $L$-approximation of
$p$, with $L=\tO(n\Gamma_{p}+b)$, in order to compute the bound $\Gamma$ for $\Gamma_{p}$, to compute the $2$-approximation $\lambda$ of $\onenorm{p}/|p_{n}|$, and to run Pan's algorithm. 
The algorithm succeeds in computing isolating disks if $b>b^{*}$ with a $b^{*}$ as in Lemma~\ref{def:certcomplexity}. Since we double $b$ in each step, we need at most $\lceil\log b^{*}\rceil$ iterations and the total cost for each iteration is bounded by (\ref{complexityresult}). This shows the complexity result.

It remains to prove the bound for $R_{i}$. When the clustering succeeds, it returns disks $D_{i}=\Delta(\tz_{i},r_{i})$ with $\min(\frac{1}{2n^{2}},\frac{\tilde{\sigma}_{i}}{256n^{2}})\le r_{i}\le \min(\frac{1}{n^2},\frac{\tilde{\sigma}_{i}}{128n^{2}})$ for all $i=1,\ldots,m$. It follows that $R_{i}=n\cdot r_{i}\le \frac{\tilde{\sigma}_{i}}{128n}$, and thus $|z_{i}-z_{j}|\ge
|\tz_{i}-\tz_{j}|-|z_{i}-\tz_{i}|-|z_{j}-\tz_{j}|>|\tz_{i}-\tz_{j}|\cdot
(1-1/(64n))>|\tz_{i}-\tz_{j}|/2$ for all $i,j$ with $i\neq j$. We conclude that
$\sigma_{i} > \tsigma_i/2 \ge 64n R_{i}$.
\ignore{\michael{You cannot use Lemma 4 since it only applies for $b\ge b_{0}$. However,when the certification step succeeds, this is not guaranteed!} }
\end{proof}

We remark that the bound (\ref{complexityresult}) can also be reformulated in
terms of values that exclusively depend on the degree $n$ and the geometry of
the roots (i.e.~their absolute values and their distances to each
other). Namely, according to Lemma~\ref{lem:tauTau}, we have $\tau_{p}\le
n+1+\log\frac{\Mea(p)}{|p_{n}|}$, and the latter expression only involves the
degree and the absolute values of the roots of $p$. This yields the bound (\ref{intro:result1})
from the introduction. 

In the next step, we show that combining our algorithm with Pan's factorization algorithm also yields a very efficient method to further refine the isolating disks.

\begin{theorem}\label{thm:refine complexity}
Let $p(x)$ be a polynomial as in Theorem~\ref{thm:isolation complexity}, and $\kappa$ be a given positive integer. We can compute isolating disks $\Delta_{i}(\tilde{z}_{i},R_{i})$ with radius $R_{i}<2^{-\kappa}$ in a number of bit operations bounded by
\begin{align}\label{complexityinteger}
\tilde{O}\left(n^{3}+n^{2}\tau_p+n\cdot\sum\nolimits_{i=1}^k \left(\log M(P_{i}^{-1})+ m_{i}\log M(\sigma_{i}^{-1})\right)+n\kappa\cdot\max_{1\le i\le k}m_{i}\right).
\end{align}
For that, we need an absolute $L$-approximation of $p$ with $L$ bounded by
$$\tO\left(n\Gamma_p+\sum\nolimits_{i=1}^k \left(\log M(P_{i}^{-1})+ m_{i}\log M(\sigma_{i}^{-1})\right)+n\kappa\cdot\max_{1\le i\le k}m_{i}\right).$$ 
\end{theorem}

\begin{proof}
	
As a first step, we use the algorithm from Section~\ref{sec:algorithm} to compute isolating disks 
$\Delta_i=\Delta(\tz_{i},R_{i})$ with $R_{i}\le \sigma_{i}/(64n)$. Each disk $\Delta_{i}$ contains the root $z_{i}$, $m_{i}=\operatorname{mult}(z_{i},p)$ approximations $\hz\in\{\hz_{1},\ldots,\hz_{n}\}$ of $z_{i}$, and it holds that $\sigma_{i}/2<\tilde{\sigma}_{i}<2\sigma_{i}$.
Let
$$\hat
{P_i}:=\prod\nolimits_{j:\hz_{j}\notin \Delta_{i}} (\tz_i-\hz_j).$$ 
We claim that $1/2 |P_{i}|<|\hat{P}_{i}|< 2|P_{i}|.$
Since 
$|\tz_i-z_i|<\sigma_i/(64n)$ for all $i$, it holds that 
$(1-\frac{1}{64n})|z_i-z_{j}|\le |\tz_i-\hz|\le (1+\frac{1}{64n})|z_i-z_j|$ for all $j\neq i$ and $\hz\in\Delta_{j}$.
Thus, $|\hat{P}_{i}|$ is a $2$-approximation of $|P_{i}|$. Similar as in the certification step, we now use approximate interval arithmetic to compute a $2$-approximation $\mu_{i}$ of $|\hat{P}_i|$, and thus a $4$-approximation of $|P_{i}|$. A completely similar argument as in the proofs of Lemma~\ref{lem:costevaluations} and Lemma~\ref{def:certcomplexity} then shows that we can compute such $\mu_{i}$'s with less than $\tilde{O}(n^{3}+n^{2}\tau_{p}+n\sum_{i}\log M(P_{i}^{-1}))$ bit operations. 
Now, from the $2$- and $4$-approximations of $\sigma_{i}$ and $|P_{i}|$, we can determine a $b_{\kappa}$ such that
\begin{itemize}
\item the properties (\ref{greaternlogn}) to (\ref{greaterr}) are fulfilled, and 
\item $2^{-b/(2m_{i})}<2^{-\kappa}$.
\end{itemize}
Then, from Corollary~\ref{thm:pansresult} and Lemma~\ref{lem:approximation quality}, we conclude that Pan's factorization algorithm (if run with $b\ge b_{\kappa}$) returns, for all $i$, $m_{i}$ approximations $\hat{z}$ of $z_{i}$ with $|\hat{z}-z_{i}|<2^{-b/(2m_{i})}<2^{-\kappa}$. Thus, for each $i$, we can simply choose an arbitrary approximation $\hat{z}\in\Delta_{i}$ and return the disk $\Delta(\hat{z},2^{-\kappa})$ which isolates $z_{i}$. The total cost splits into the cost for the initial root isolation and the cost for running Pan's Algorithm with $b=b_{\kappa}$. Since the latter cost is bounded by $\tilde{O}(n b_{\kappa}+n^{2}\Gamma_{p})$, the bound (\ref{complexityinteger}) follows.
	\end{proof}

Finally, we apply the above results to the important special case, where we aim to isolate the roots of a polynomial with integer coefficients.

\begin{theorem}\label{thm:square free complexity}
Let $p(x)\in\Z[x]$ be a polynomial of degree $n$ with integer coefficients of size less than $2^{\tau}$. Then, we can compute isolating disks $\Delta(\tz_{i},R_{i})$, with $R_{i} < \frac{\sigma_{i}}{64n}$, for all roots $z_{i}$ together with the corresponding multiplicities $m_{i}$ using
\begin{align}\label{complexityintegerA}
\tilde{O}(n^{3}+n^{2}\tau)
\end{align}
bit operations.
For a given positive integer $\kappa$, we can further refine the disks $\Delta_{i}$ to a size of less than $2^{-\kappa}$ with a number of bit operations bounded by
\begin{align}\label{complexityintegerB}
\tilde{O}(n^{3}+n^{2}\tau+n\kappa).
\end{align}
\end{theorem}

	\begin{proof}
	In a first step, we compute the square-free part $p^{*}=p/\gcd(p,p')$ of $p$. 
	According to~\cite[\S 11.2]{Gathen-Gerhard:book}, we need $\tilde{O}(n^{2}\tau)$ bit 
	operations for this step, and $p^{*}$ has integer coefficients of bitsize $O(n+\tau)$. The degree of 
	$p^{*}$ yields the number $k$ of distinct roots of $p$. In order to use our root isolation algorithm from Section~\ref{sec:algorithm}, we divide $p$ by its leading coefficients $p_n$ to meet the requirement that the leading coefficient has absolute value in $[1/4,1]$. Obviously, the roots are not affected by this normalization step.	
	
	 Now, in order to derive the bound in (\ref{complexityintegerA}), we have to reformulate the bound from (\ref{complexityresult}) in terms of the 
	degree $n$ and the bitsize $\tau$ of $p$. We first use~\cite[Theorem 2]{es-bisolvecomplexity-11} to show that $\sum_{i=1}^k m_{i}\log\max(1,\sigma_{i}^{-1})=\tilde{O}(n^{2}+n\tau)$. Furthermore, we 
	have $\tau_{p}\le \tau$. Hence, it remains to show that 
	$n\cdot\sum_{i=1}^k \log M(P_{i}^{-1})=\tilde{O}(n^{3}+n^{2}\tau)$. For that, we consider a square-free factorization $
	p(x)=\prod_{l=1}^{n} (Q_{l}(x))^{l}
	$
	with square-free polynomials $Q_{l}\in\Z[x]$ such that $Q_{l}$ and $p/Q_{l}^{l}$ are coprime 
	for all $l=1,\ldots,n$. Note that the roots of 
	$Q_{l}$ are exactly the roots of $p$ with multiplicity $l$, and that $Q_{l}$ is a constant for most $l$. We further denote $\bar{p}:=p/\lc(p)$ 
	and $\bar{Q}_{l}:=Q_{l}/\lc(Q_{l})$. Let $S_{l}$ denote the set of
        roots of $Q_{l}$. Then, from the definition of $P_{i}$,
	\begin{align*}
	\prod_{i:z_{i}\in S_{l}} \vert{P_{i}|}&=\prod_{i:z_{i}\in S_{l}}\prod_{j\neq i} |z_{i}-z_{j}|^{m_{j}}\\
	&=\prod_{i:z_{i}\in S_{l}}\prod_{j\neq i : z_{j}\notin S_{l}}|z_{i}-z_{j}|^{m_{j}}\cdot 	\prod_{i:z_{i}\in S_{l}}\prod_{j\neq i : z_{j}\in S_{l}}|z_{i}-z_{j}|^{l}\\
        &=\prod_{i\in S_{l}} |(\bar{p}/\bar{Q}_{l}^{l})(z_{i})|\cdot \prod_{i:z_{i}\in S_{l}} |(\bar{Q}_{l})'(z_{i})|^{l}\\
        &=|\operatorname{res}(\bar{p}/\bar{Q}_{l}^{l},\bar{Q}_{l})|\cdot |\operatorname{res}(\bar{Q}_{l},(\bar{Q}_{l})')|^{l}\\
        &=\frac{|\operatorname{res}(p/Q_{l}^{l},Q_{l})|}{|\lc(Q_{l})^{n-l\cdot\deg Q_{l}}(\lc(p/Q_{l}^{l}))^{\deg Q_{l}}|}\cdot \left\vert\frac{\operatorname{res}(Q_{l},Q_{l}')}{\lc(Q_{l})^{2\deg Q_{l}-1}(\deg Q_{l})^{\deg Q_{l}}}\right\vert^{l}\\
        &\ge \frac{1}{|\lc(Q_{l})|^{n-l}\cdot|\lc(p)|^{\deg Q_{l}}\cdot n^{l\deg Q_{l}}}
        	\end{align*}
where $\operatorname{res}(f,g)$ denotes the
resultant\footnote{For univariate polynomials 
\begin{align*}
\operatorname{res}(f,g) &= \lc(f)^{\deg{g}} \lc(g)^{\deg{f}}
  \prod_{(x,y): f(x) = g(y) = 0} (x - y) = \lc(f)^{\deg{g}} 
  \prod_{x: f(x) = 0} g(x).\end{align*}} of two polynomials $f$ and $g$. For the last inequality, we used that $\operatorname{res}(p/Q_{l}^{l},Q_{l})\in\Z$ and $\operatorname{res}(Q_{l},Q_{l}')\in\Z$. Taking the product over all $l$ yields
\begin{align*}
\prod_{i=1}^k|P_{i}|&\ge \left\vert\frac{1}{\prod_{l=1}^{n}(\lc(Q_{l})^{n-l}\cdot\lc(p)^{\deg Q_{l}}\cdot n^{l\deg Q_{l}})}\right\vert\ge\frac{1}{|\lc(p)|^{2n}\cdot n^{n}}\ge 2^{-2n\tau-n\log n}.
\end{align*}
Note that, for any $i$, we also have 
$$|P_{i}|=\frac{|p^{(m_{i})}(z_{i})|}{m_{i}!p_{n}}<\frac{m_{i}! 2^{\tau}(n+1) M(z_{i})^{n}}{m_{i}! |p_{n}|}\le n 2^{\tau+1}M(z_{i})^{n},$$
and, thus,
\[
\sum_{i=1}^k\log M(P_{i}^{-1})=\tilde{O}(n\tau+n\sum_{i=1}^k \log M(z_{i}))=\tilde{O}(n\tau),
\]
where we used that $\sum_{i}\log M(z_{i}))\le \log \operatorname{Mea}(p)\le \log\onenorm{p}<\log (n+1)+\tau$. This shows (\ref{complexityintegerA}).

For the bound in (\ref{complexityintegerB}) for the cost of refining the isolating disks $\Delta_{i}(\tz_{i},R_{i})$ to a size of less than $2^{-\kappa}$, we consider the square-free part $p^{*}$. Note that the disks $\Delta_{i}$ obtained in the first step are obviously also isolating for the roots of $p^{*}$ ($p$ and $p^{*}$ have exactly the same distinct roots) and that $R_{i}<\sigma(z_{i},p)/(64n)=\sigma(z_{i},p^{*})/(64n)\le \sigma(z_{i},p^{*})/(64\deg p^{*})$. Thus, proceeding in completely analogous manner as in the proof of Theorem~\ref{thm:refine complexity} (with the square-free part $p^{*}$ instead of $p$) shows that we need $\tilde{O}(n^{3}+n^{2}\tau+n\kappa)$ bit operations for the refinement. This proves the second claim.
\end{proof}

\subsection{Well-separated Clusters of Roots}\label{sec: well-separated clusters}

\Kurt{We now turn to the problem of computing well-separated clusters of roots. We no longer insist that the clusters are in one-to-one correspondence with the roots, but may have clusters containing more than one root as long as the clusters are well-separated. Well-separated means that the diameter of each cluster is much smaller than the distance from the cluster to the nearest distinct cluster. We also need to impose an upper bound on the diameter of any cluster to make the problem non-trivial. Otherwise, it would be allowed to return a single cluster, e.g., the disk centered at the origin and having radius $4 \max_i \abs{p_i}/\abs{p_n}$, containing all roots. Recall that this disk contains all roots of $p$ (Lemma~\ref{lem:tauTau}). }

\Kurt{Renegar's algorithm~\cite{Renegar87} computes clusters of radius $\epsilon$, where $\epsilon > 0$ is an input parameter. More precisely, it computes $\tz_1$ to $\tz_j$ (the number of clusters is not predetermined) and multiplicities $m_i$ such that $\sum_i m_i = n$, the disks $\Delta_{i}=\Delta(\tz_i,\epsilon)$ are disjoint, and $\Delta_i$ contains exactly $m_i$ roots of $p$. He uses subdivision and Newton iteration for root approximation, the Shur-Cohn method~\cite[Theorem 6.8b]{Henrici74} for determining whether a disk contains a root, and an approximate winding number algorithm for estimating the number of zeros in a disk. The arithmetic complexity (= number of arithmetic operations) is analyzed and shown to be nearly optimal. The author also states that ``his algorithm will not fare well in the bit-complexity model''.  }

\Kurt{Yakoubsohn and Giusti~et.~al.~\cite{Yakoubsohn00,Giusti-Lecerf-Salvy-Yakoubsohn05} show how to approximate a single cluster of zeros. Given a good starting point, they derive an estimate for the number of zeros in the cluster from the convergence rate of Newton's method. They verify the number of roots in a cluster by an inclusion test based on Rouch\'{e}'s theorem. Schr\"{o}der's variant of Newton's method is used to improve the approximation of the cluster. In the case of a multiple root of known multiplicity it is known to converge quadratically~\cite{Henrici74}. They show that, in the case of a cluster of roots, it is still quadratic provided the iteration is stopped sufficiently early. They propose a method for stopping the iteration at a distance from the cluster which is on the order of its diameter.}

\Kurt{We modify our algorithm as follows. The input to the algorithm is the polynomial $p$. In the clustering algorithm (Section~\ref{sec:clustering}), we drop step (6), i.e., we allow the algorithm to generate any number of clusters. After the clustering, we proceed to the verification step. If the verification step succeeds (this includes a check that the disks have radius at most $4 \max_i \abs{p_i}/\abs{p_n}$), we output the clusters determined in the clustering step. Otherwise, we double $b$ and repeat. The modified algorithm has the following properties:
\begin{enumerate}
\item If it returns disks $D_1$ to $D_j$ and associated multiplicities $m_1$ to $m_j$, then $\sum_i m_j = n$, $\Delta_j$ contains $m_j$ root approximations and $m_j$ roots of $p$ counted with multiplicity, and the disks with the $n$-fold radii are pairwise disjoint. 
\item The algorithms stops at the latest when the precision exceeds $b_0$, where $b_0$ is as in the preceding section.
\item The bit complexity of the algorithm is as stated in (\ref{complexityresult}).
\end{enumerate}}

\section{Curve Analysis}\label{sec:curveana}

\Kurt{In this section, we show how to integrate our approach to isolate and approximate the roots of a univariate polynomial in an algorithm to compute a cylindrical algebraic decomposition~\cite{Collins:CAD,bpr-arag-06,beks:top2D,Gonzales-Vega,cjk-iaicad-02,
Hong,Strzebonski,Eigenwillig-Kerber-Wolpert,ks-top-12,Cheng-et-al}.
More specifically, we apply the results from the previous section
to a recent algorithm, denoted~\topo, from~\cite{beks:top2D}
for computing the topology of a real planar algebraic curve. This yield a bound
on the expected number of bit operations for
computing the topology of a real planar algebraic curve that improves the currently best
bound~\cite{ks-top-12} from
$\tO(n^{9}\tau+n^{8}\tau^{2})$ (deterministic) to $\tO(n^{6}+n^{5}\tau)$
(randomized). Isolating the real-valued solutions of a bivariate polynomial system $g(x,y)=h(x,y)=0$
can be reduced to the problem of computing the topology of an algebraic curve of a degree comparable to the degree of the polynomials $g$ and $h$.
Based on the latter observation, we derive a bound on the expected number of bit operations for solving a bivariate polynomial system that improves the best known bound~\cite{es-bisolvecomplexity-11} from $\tO(n^{8}+n^{7}\tau)$ (deterministic) to $\tO(n^6+n^5\tau)$; see Theorem~\ref{complexity:systemsolving}.}

We also remark that an implementation of algorithm~\topo\ is available~\cite{beks:top2D}. The implementation uses a variant of the Aberth-Ehrlich method for root isolation~\cite{KobelMasterThesis,Rump:TenMethods} and shows great efficiency in practice.

\subsection{Review of the Algorithm \textsc{TopNT}}\label{sec:algotopology}

For the sake of a self-contained representation, we briefly review the algorithm \textsc{TopNT}. For more details and the corresponding proofs, we refer to~\cite{beks:top2D}. The input of the algorithm is a bivariate polynomial $f\in\Z[x,y]$ 
of total degree $n$ with integer coefficients of magnitude $2^\tau$ or less. The polynomial defines an algebraic curve 
$$
C:=\{(x,y)\in\C^2:f(x,y)=0\}\subseteq \C^2.
$$
The algorithm 
 returns a 
planar straight-line graph $\mathcal{G}$ embedded in $\mathbb{R}^2$ that is
\emph{isotopic\footnote{We actually consider the stronger notion of an \emph{ambient isotopy}, but omit the ``ambient''. $\mathcal{G}$ is ambient isotopic to $C_{\R}$ if there is a continuous mapping $\phi:[0,1]\times \R^{2}\mapsto \R^2$ with
$\phi(0,\cdot)=\operatorname{id}_{\R^{2}}$, $\phi(1,C_{\R})=\mathcal{G}$, and
$\phi(t_0,\cdot)$ is a homeomorphism for each $t_0\in[0,1]$.} to the real part}
$C_{\R}:=C\cap \R^{2}$ of $C$.

In the first step (the \textbf{shearing step}), we choose an $s\in\Z$ at random (initially, consider $s=0$) and consider the sheared curve $$C_s:=\{(x,y)\in\C^2: f_s(x,y):=f(x+s\cdot y,y)=0\}.$$ Then, any planar graph isotopic to the real part $C_{s,\R}:=C_s\cap\R$ of $C_s$ is also isotopic to $C_{\R}$, and vice versa. We choose $s$ such that the leading coefficient $\operatorname{lcf}(f_s(x,y);y)$ (with respect to $y$) of the defining polynomial $f_s(x,y)$ of $C_s$ is a constant. This guarantees that $C_{s,\R}$ has no vertical asymptote and that it contains no vertical line. By abuse of notation, we write $C=C_s$ and $f=f_s$ throughout the following considerations.

In the \textbf{projection step}, the $x$-critical points of $C$ (i.e.~all points $(x_0,y_0)\in C$ with $f_y(x_0,y_0)=0$, where $f_y:=\frac{\partial f}{\partial y}$) are projected onto the real $x$-axis by means of a resultant computation. More precisely, we compute\medskip 
\begin{itemize}
\item $R:=\operatorname{res}(f,f_y;y)\in\Z[x]$, 
\item its square-free part $R^*:=R/\gcd(R,R')$, 
\item isolating intervals $I_1,\ldots,I_m$ for the real roots $\alpha_1,\ldots,\alpha_m$ of $R^*$, 
\item the multiplicity $m_i:=\operatorname{mult}(\alpha_i,R)$ of $\alpha_i$ as a root of $R$ for all $i=1,\ldots,m$, and
\item arbitrary separating values $\beta_0,\ldots,\beta_{m+1}\in\R$ with $\alpha_m<\beta_{m+1}$, and $\beta_{i-1}<\alpha_i<\beta_i$ for all $i=1,\ldots,m$.\medskip
\end{itemize}
We further compute\medskip 
\begin{itemize}
\item $f_x^*:=\frac{f_x}{\gcd(f_x,f_y)}$ and $f_y^*:=\frac{f_y}{\gcd(f_x,f_y)}$,
\item $Q:=\operatorname{res}(f_x^*,f_y^*;y)$, and
\item the multiplicity $l_i:=\operatorname{mult}(\alpha_i,Q)$ of $\alpha_i$ as a root of $Q$ for all $i=1,\ldots,m$.\medskip
\end{itemize}

In the \textbf{lifting step}, we compute the fibers of $C$ at the points
$\alpha_i$ and $\beta_i$, that is, we isolate the roots of the polynomials
$f_{\alpha_i}(y):=f(\alpha_i,y)\in\R[y]$ and
$f_{\beta_i}(y):=f(\beta_i,y)\in\R[y]$. For that, we first compute the number
of distinct complex roots of each of these polynomials, and then use the
root isolator from Section~\ref{sec:rootisolation}.\footnote{More precisely, we first compute some $2^t$, with $\operatorname{lcf}(f_s(x,y);y)\le 2^t \le 4\cdot\operatorname{lcf}(f_s(x,y);y)$, and apply the root isolator from Section~\ref{sec:rootisolation} to the polynomial $2^{-t}\cdot f_{\alpha_i}(y)$ (and $2^{-t}\cdot f_{\beta_i}(y)$, respectively) which has leading coefficient of absolute value between $1/4$ and $1$.}
Obviously, each polynomial $f_{\beta_i}(y)$ has $k(\beta_i)=\deg f_{\beta_i}=n$
distinct complex roots. The difficult part is to determine the number
$k(\alpha)$ of distinct roots of $f_\alpha(y)$ for a root $\alpha$ of $R^*$.
According to~\cite[(3.6)]{beks:top2D} and~\cite[Theorem 5]{beks:top2D}, 
\begin{equation}\label{eq:teissier}
k^+(\alpha):=n-\operatorname{mult}(\alpha,R)+\operatorname{mult}(\alpha,Q)\ge k(\alpha),
\end{equation}
and, for a generic 
shearing factor $s$ (more precisely, for all but $n^{O(1)}$ many $s$), the equality $k^+(\alpha)=k(\alpha)$ holds for all roots $\alpha$ of $R^*$. Summation over all complex roots of $R^*$ then yields 
\begin{align*}
& K^+:=\sum_{\alpha:R^*(\alpha)=0} k^+(\alpha)=n\cdot \deg R^* - \deg R + \deg \gcd (R^\infty, Q)\ge \sum_{\alpha:R^*(\alpha)=0} k(\alpha)=:K,\text{ and}\\
& K=K^+\text{ for generic }s,
\end{align*}
where $\gcd (R^\infty, Q)$ is defined as the product of all common factors of $R$ and $Q$ with multiplicities according to their occurrence in 
$Q$. The crucial idea is now to compare the upper bound $K^+$ with a lower bound $K^-$ which also equals $K$ up to a non-generic choice of some parameters. In order to understand the computation of $K^-$, we first consider the exact computation of $K$:
Let  $\operatorname{Sres}_i(f, f_y; y)\in\Z[x,y]$ denote the \emph{$i$-th subresultant polynomial} of $f$ and $f_y$ (with respect to $y$), 
and $\operatorname{sr}_i(x):=\operatorname{sres}_i(f,f_y;y) \in \Z[x]$ its leading coefficient. In particular, we have $R=
\operatorname{sres}_0(f,g;y)=\operatorname{res}(f,f_y;y)$. We define:
\begin{alignat}{3}\label{def:Ris}
  S_0 &:= R^*,
  &\text{ }
  S_i &:= \gcd (S_{i-1}, \operatorname{sr}_i) &\\
  R_1 &:= \frac{S_0}{S_1} = \frac{S_0}{\gcd(S_0, S_1)},
  &\text{ }
  R_i &:= \frac{S_{i-1}}{S_i} = \frac{\gcd(S_0, \dots, S_{i-1})}{\gcd(S_0, \dots, S_{i-1},
    S_i)} ,&\nonumber
\end{alignat}
where $i=1,\ldots,n$. Then, $\prod_{i\ge 1} R_i$ constitutes a factorization
of $R^*$ such that $R_i(\alpha)=0$ if and only if $f(\alpha,y)$ has exactly $n-i$ distinct complex roots; see~\cite[Section 3.2.2]{beks:top2D} for details.
Hence, we have $K=\sum_{i\ge 1} (n-i)\cdot \deg R_i$. We do not carry out the latter computation of $K$ over the integer domain, but over a modular prime field which yields a lower bound $K^-$ for $K$.
\ignore{Unfortunately, this computation of $K$ is costly in practice\marginpar{why the qualifier in practice?} due to the 
computation \Kurttwo{of
the full subresultant sequence.}{in (\ref{def:Ris})} \marginpar{but, we compute the full sequence in Lemma 15}Instead, in order to derive a lower bound $K^-$, we do not
carry out the above computations directly in $\mathbb{Z}$, but in a modular prime field.}More precisely, we choose
a prime $p$ at random, compute the modular images 
$\operatorname{sr}_i^{( p)}(x)=\operatorname{sres}_i(f\operatorname{mod} p,f_y\operatorname{mod} p;y)\in\Z_p[x]$ of 
$\operatorname{sr}_i(x)\in\Z[x]$, and perform all
computations from (\ref{def:Ris}) in $\Z_p[x]$. This yields polynomials $R_i^{( p)}\in\Z_p[x]$. Now,~\cite[Lemma 4]{beks:top2D} shows that
\begin{equation}\label{eq:Kminus}
K^-:=\sum_{i\ge 1} (n-i)\cdot \deg R_i^{( p)}\le K,
\end{equation}
and $K^-=K$ for all but finitely many bad primes.\footnote{In the computation of the $S_{i}$'s and $R_{i}$'s (over $\mathbb{Z}$), all intermediate results have integer coefficients of bitsize bounded by $(n\tau)^{O(1)}$. Since the product of $N$ distinct primes is larger than $N!=2^{\Omega(N\log N)}$, there exist at most $(n\tau)^{O(1)}$ many bad primes for which $K^{-}\neq K$.} Hence, if $K^-<K^+$, we have either chosen a bad prime or a bad shearing
factor. In this case, we start over with a new $s$ and choose a new prime $p$ in the lifting step. If $K^-=K^+$, we know for
sure that $K^+=K$, and thus $k^+(\alpha)=k(\alpha)$ for all roots $\alpha$ of the resultant polynomial $R$.

We can now use our method from Section~\ref{sec:rootisolation} to isolate all complex roots of the fiber polynomials $f_{\alpha_i}(y)$ and $f_{\beta_i}(y)$.
Namely, we can ask for arbitrary good approximations of $\alpha_i$ and $\beta_i$ (by refining corresponding isolating intervals), and thus 
for arbitrary good approximations of the coefficients of the fiber polynomials. In addition, we know the exact number of distinct roots of
either polynomial. From the isolating regions in $\C$, we then derive isolating intervals for the real roots together with corresponding multiplicities. If one of the polynomials $f_{\alpha_i}(y)$ has more than one multiple real root, we start over and choose a new shearing 
factor $s$. Otherwise, we proceed with the final step.\\ 

\textbf{Connection step}. We remark that, except for finitely many $s$, each $f_{\alpha_i}$ has exactly one multiple root.
The previous two steps already yield the vertices of the graph $\mathcal{G}$. Namely, these are exactly the points\footnote{For a graph with rational vertices, you may replace each $x_0=\alpha_i$ (or $x_0=\beta_i$) by an arbitrary rational value in its corresponding isolating interval, and the same for each real root of $f_{x_0}(y)$.}
\[
V(\mathcal{G}):=\{(x,y)\in\R^2: \exists i\text{ with } x=\alpha_i\text{ or } x=\beta_i,\text{ and }f(x,y)=0\}
\]
Since each polynomial $f_{\alpha}(y)$ has exactly one multiple root, there exists a unique vertex $v$ along each vertical line, where either the number of edges connecting $v$ to the left or to the right may differ from one. Hence, connecting all vertices
in an appropriate manner is straightforward; see~\cite[Section 3.2.3]{beks:top2D} for more details.\\

\Kurt{\emph{Remark.} We remark that we use randomization at exactly two stages of the algorithm, that is, the choice of a shearing 
value $s$ in the projection step and the choice of a prime $p$ for computing the lower bound $K^-$ for $K$ in the lifting 
step. Let $P$ denote the set of all prime numbers, then there exists a set $B\subset \Z\times P$ of "bad" pairs $(s,p)\in B$ for 
which success of the algorithm is not guaranteed, whereas, the algorithm returns the correct topology of $C$ for all other pairs. 
There are at most $n^{O(1)}$ "bad" values for $s$ that yield a non-generic position of the curve, and, for 
each of the remaining values for $s$, there exist at most $(n\tau)^{O(1)}$ many "bad" choices for $p$. Since we can generate a random prime of bit length $L$ or less for the cost of $L^{O(1)}$ bit operations,\footnote{In order to generate a random prime 
of size less than $2^L$, pick an integer of magnitude less than $2^L$ at random and test this integer for being prime. Since 
the cost for the latter test is polynomial~\cite{PrimeP} in $L$ and since there exist~\cite{primenumbers} more than 
$2^L/(\operatorname{ln}(2)\cdot L+2)$ prime numbers of size less than $2^{L}$ for any $L\ge 6$, we can pick a random prime of bit length less than 
$L$ with a number of bit operations that is polynomial in $L$.} it follows that using 
$(\log(n\tau))^{O(1)}$ bit operations, we can pick a pair $(s,n)$ such that, with probability $1/2$, the algorithm succeeds.}

\subsection{Complexity Analysis}\label{sec:curvecomplexity}

Throughout the following considerations, we say that a polynomial $G\in\Z[x_{1},\ldots,x_{k}]$ with integer coefficients has \emph{magnitude} $(N,\mu)$ if the total degree of $G$ is upper bounded by $N$ and all coefficients have absolute value $2^{\mu}$ or less. 
In addition, we fix the following notations: For an arbitrary $\alpha\in\C$,\medskip
\begin{itemize}
\item we define $f_{\alpha}(y):=\sum_{i=0}^{n} f_{\alpha,i}y^{i}:=f(\alpha,y) \in\C[y]$, where $f\in\Z[x,y]$ is our input polynomial. We further define $\tau_{\alpha}:=\log \max_{i} |f_{\alpha,i}|\ge 0$ (notice that, for the considered shearing factors $s$, the leading coefficient $f_{\alpha,n}$ is a constant integer for all $\alpha$).
\item the number of distinct roots of $f_{\alpha}$ is denoted by $k(\alpha)$. We further denote 
$z_{\alpha,1},\ldots,z_{\alpha,k(\alpha)}$ the distinct roots of $f_{\alpha}$, and $m_{\alpha,i}$, with $i=1,\ldots,k(\alpha)$, the corresponding multiplicities.
\item $\sigma_{\alpha,i}$ denotes the separation of $z_{\alpha,i}$, and $P_{\alpha,i}:=\prod_{j\neq i} (z_{\alpha,i}-z_{\alpha,j})^{m_{\alpha,j}}$.
\item For an arbitrary polynomial $G\in\C[x]$, we denote $V(G)$ the set of all distinct complex roots of $G$, and $\mathcal{V}(G)$ the multiset of all complex roots (i.e.~each root occurs a number of times according to its multiplicity).\medskip 
\end{itemize}

We first prove the a couple of basic results which are needed for our analysis:

\begin{lemma}\label{lem:complexitydiv}
For a fixed positive integer $k$, let $G\in\Z[x_{1},\ldots,x_{k}]$ be an integer polynomial of 
magnitude $(N,\mu)$. Then, each divisor $g\in\Z[x_{1},\ldots,x_{k}]$ of $G$ has 
coefficients of bitsize $\tilde{O}(\mu+N)$. 
\end{lemma}

\begin{proof}
We prove the claim via induction over $k$. For a univariate $G\in\Z[x_{1}]$, we remark that $\Mea(g)\le\Mea(G)\le \norm{G}_{2}\le 2^{\mu+1}N$, and thus the absolute value of each coefficient of $g$ is bounded by $2^{N}\Mea(g)\le 2^{N+\mu+1}N$. 

For the general case, we write $$g(x_{1},\ldots,x_{k})=\sum_{\lambda=(\lambda_{1},\ldots,\lambda_{k-1})}a_{\lambda}(x_{k})x_{1}^{\lambda_{1}}\cdots x_{k-1}^{\lambda_{k-1}}, \text{ with }a_{\lambda}\in\Z[x_{k}].$$
For a fixed $\bar{x}_{k}\in\{0,\ldots,N\}$, the polynomial $g(x_{1},\ldots,x_{k-1},\bar{x}_{k})\in\Z[x_{1},\ldots,x_{k-1}]$ is a divisor of $G(x_{1},\ldots,x_{k-1},\bar{x}_{k})\in\Z[x_{1},\ldots,x_{k-1}]$. Since $|\bar{x}_{k}|^{N}\le N^{N}=2^{N\log N}$ and $a_{\lambda}(x_{k})$ has degree $N$ or less, it follows that $G(x_{1},\ldots,x_{k-1},\bar{x}_{k})$ has bitsize $O(N\log N+\mu)$. Hence, from the induction hypothesis, we conclude that the polynomial $g(x_{1},\ldots,x_{k-1},\bar{x}_{k})$ has coefficients of bitsize $\tilde{O}(\mu+N)$, and thus $a_{\lambda}(i)\in\Z$ has bitsize $\tilde{O}(\mu+N)$ for all $i=0,\ldots,N$ and all $\lambda$. Since $a_{\lambda}$ is a polynomial of degree at most $N$, it follows that $a_{\lambda}$ is uniquely determined by the values $a_{\lambda}(i)$, and thus Lagrange interpolation yields
\[
a_{\lambda}(x)=\sum_{i=0}^{N}a_{\lambda}(i)\cdot\frac{x\cdot (x-1)\cdots (x-i+1)(x-i-1)\cdots (x-N)}{i\cdot (i-1)\cdots 1\cdot (-1)\cdots (i-N)}
\]
Expanding the numerator of the fraction yields a polynomial with coefficients of absolute value $2^{O(N\log N)}$, and thus each coefficient of $a_{\lambda}(x_{k})$ has bitsize $\tilde{O}(\mu+N)$ because $a_{\lambda}(i)$ has bitsize $\tilde{O}(\mu+N)$ and there are $N+1$ summands. This proves the claim.
\end{proof}

In addition, we need a bound on the bit complexity of computing the greatest common divisor of two univariate polynomials with integer coefficients. For a proof of the following result, we refer to~\cite[\S 11.2]{Gathen-Gerhard:book}.

\begin{lemma}\label{lem:gcdunivariate}
Let $F,G\in\Z[x]$ be two univariate polynomials of magnitude $(N,\mu)$. 
\begin{itemize}
\item Computing $H:=\gcd(F,G)$ uses $\tilde{O}(N^2\mu)$ bit operations. 
\item Given a polynomial $H\in\Z[x]$ that divides $F$, computing $F/H$ uses $\tilde{O}(N\mu)$ bit operations.
\end{itemize} 
\end{lemma}

In the projection step of \textsc{TopNT}, we also have to compute the greatest common divisor of two bivariate polynomials. Although it is not very difficult to derive a reasonable good bound on the bit complexity of the latter problem, it seems that no complexity results for \emph{deterministic} algorithms are published so far. The following lemma provides such a result:

\begin{lemma}\label{lem:gcdbivariate}
Let $F,G\in\Z[x,y]$ be two bivariate polynomials of magnitude $(N,\mu)$. Then, we can compute $H(x,y):=\gcd(F,G)$ with $\tilde{O}(N^6+N^5\mu)$ bit operations.
\end{lemma} 

\begin{proof}
Throughout the following considerations, we say that two polynomials $p_1,p_2$ in $\Z[x]$ (or in $\Z[x,y]$) are equivalent (written as $p_1\simeq p_2$) if there exists an integer $\lambda\in\Z$ such that $p_1=\lambda\cdot p_2$ or $p_2=\lambda\cdot p_1$.
Let $F(x,y)=f_l(x)\cdot y^l+\cdots+f_0(x)$ and $G(x,y)=g_k(x)\cdot y^k+\cdots+g_0(x)$, with $l,k\le N$ and polynomials $f_i,g_j\in\Z[x]$. For an arbitrary but fixed $s\in\Z$, it holds that 
$$
H(x,y)=\gcd(F,G)=\hat{H}(x-s\cdot y,y),\text{ where }\hat{H}(x,y):=\gcd(F(x+s\cdot y,y),G(x+s\cdot y,y)). 
$$
Namely, for each divisor $d(x,y)\in\Z[x,y]$ of $F$ and $G$, $d(x\pm s\cdot y,y)$ is also a divisor of $F(x\pm s\cdot y,y)$ and $G(x\pm s\cdot y,y)$, respectively. Thus, for computing $H(x,y)$, it suffices to compute $\hat{H}$ for an arbitrary integer $s_0\in\Z$ and to replace $x$ by $x-s_0\cdot y$. We first determine an integer $s_0$ such that both polynomials $\hat{F}:=F(x+s_0 y,y)$ and $\hat{G}:=G(x+s_0y,y)$ have  constant leading coefficients with respect to $y$. That is,
\begin{align}\label{properties:hatf}
&\hat{F}(x,y)=\hat{f}_m(x)\cdot y^m+\cdots+\hat{f}_0(x)\text{ and }\hat{G}(x,y)=\hat{g}_n(x)\cdot y^n+\cdots+\hat{g}_0(x),\\
&\nonumber\text{ with }\hat{f}_i,\hat{g}_j\in\Z[x]\text{ and }\hat{f}_m,\hat{g}_n\in\Z.
\end{align}
Considering $s$ as an indeterminate variable, computing $F(x+s\cdot y,y)\in\Z[s,x,y]$ and $G(x+s\cdot y,y)\in\Z[s,x,y]$ uses $\tilde{O}(N^4+N^3\mu)$ bit operations because computing $(x+sy)^i$ for all $i=0,\ldots,N$, needs $\tilde{O}(N^3)$ bit operations, and computing $f_i(x)\cdot (x+sy)^i$ and $g_i(x)\cdot (x+sy)^i$ for a fixed $i$ needs $\tilde{O}(N^2(N+\mu))$ bit operations. The leading coefficients of $F(x+sy,y)$ and $G(x+sy,y)$ with respect to $y$ are univariate polynomials in $s$ of magnitude $(N,O(N\log N+\mu))$. Thus, computing an integer $s_0$, with $|s_0|\le N$, such that both of the latter univariate polynomials do not vanish needs at most $\tilde{O}(N^2\mu+N^3)$ bit operations (polynomial evaluation at the $2N+1$ points $s=-N,-N+1,\ldots,-1,0,1,\ldots,N$). It follows that computing an $s_0$, with $|s_0|\le N$, which fulfills the desired properties from (\ref{properties:hatf}) needs $\tilde{O}(N^4+N^3\mu)$ bit operations. Throughout the following considerations, we can further assume that there exists no integer different from $\pm 1$ which divides $\hat{F}$ and $\hat{G}$. Namely, with $\tilde{O}(N^2(N+\mu))$ bit operations, we can divide 
$\hat{F}$ and $\hat{G}$ by the greatest common divisor $\lambda$ (which is an integer because of $\hat{f}_m,\hat{g}_n\in\Z$) of all coefficients $\hat{f}_i$ and $\hat{g}_j$.

We now come to the computation of $\hat{H}(x,y)=\gcd(\hat{F},\hat{G})$. According to~\cite{reischert-asymptotic,Diochnos:2009}, we can compute the subresultant sequence $\operatorname{Sres}_i(\hat{F}, \hat{G}; y)\in\Z[x,y]$ with $\tO(N^{6}+N^5\mu)$ bit operations since the polynomials $\hat{F}$ and $\hat{G}$ have magnitude $(N,O(\mu+N\log N))$. The total degree of each polynomial $\operatorname{Sres}_i(\hat{F}, \hat{G}; y)$ is bounded by $N^2$, the $y$-degree is bounded by $N-i$, and all coefficients have bitsize $\tO(N(N+\mu))$. 
Let $i_0$ be the smallest index with $\operatorname{Sres}_{i_0}(\hat{F}, \hat{G}; y)\not\equiv 0$, then $$\bar{H}(x,y):=\operatorname{Sres}_{i_0}(\hat{F}, \hat{G}; y)=\bar{h}_{i_0}(x)\cdot y^{i_0}+\cdots+\bar{h}_0(x)=\sum_{i,j}\bar{h}_{ij}\cdot x^i y^j$$ coincides with $\hat{H}$ up to a fraction $\frac{p(x)}{q(x)}$ with coprime $p,q\in\Z[x]$. That is,
\[
\hat{H}(x,y)=\frac{p(x)}{q(x)}\cdot \bar{H}(x,y)=\left(\frac{p(x)}{q(x)}\cdot \bar{h}_{i_0}(x)\right)\cdot y^{i_0}+\cdots+\left(\frac{p(x)}{q(x)}\cdot \bar{h}_0(x)\right).
\]
Again, we can assume that there exists no integer different from $\pm 1$ that divides all coefficients $\bar{h}_{ij}$ of $\bar{h}$. Namely, we can divide $\bar{H}$ by the greatest common divisor of all $\bar{h}_{ij}$, and this computation uses $\tilde{O}(N^2\cdot N(N+\mu))$ bit operations.
Since the leading coefficients of $\hat{F}$ and $\hat{G}$ with respect to $y$ are constants, the same also holds for $\hat{H}=\gcd(\hat{F},\hat{G})$. Thus, $p$ must be a constant and $q\simeq \bar{h}_{i_0}$.
It follows that the primitive part\footnote{The primitive part $P^*$ of a polynomial $P(x)=p_N\cdot x^N+\cdots+p_0\in\Z[x]$ is defined as $P^*(x):=\gcd(p_N,\ldots,p_0)^{-1}\cdot P(x)$.} $\bar{h}_{i_0}^*$ of $\bar{h}_{i_0}$ divides all coefficients of $\bar{H}$, and that $\bar{H}/\bar{h}_{i_0}^*\simeq \hat{H}$. Hence, since $\hat{H}$ is primitive, we must have $\bar{H}/\bar{h}_{i_0}^*=\hat{H}$.
The computation of $h_{i_0}^*$ and $\bar{H}/\bar{h}_{i_0}^*$ needs $\tilde{O}(N^6+N^5\mu)$ bit operations, where we use the fact that the polynomials $\bar{h}_i(x)$ have magnitude $(N^2,\tilde{O}(N(N+\mu)))$ and that each division of $\bar{h}_i(x)$ by $h_{i_0}^*$ is remainder-free; cf. Lemma~\ref{lem:gcdunivariate}. Finally, computing $H=\gcd(F,G)$ from $\hat{H}=H(x+s_0\cdot y,y)$ uses $\tilde{O}(N^4+N^3\mu)$ bit operations since $\hat{H}$ has magnitude $(N,O(N\log N+N\cdot \log s_0+\mu))=(N,O(N\log N+\mu))$. 
\end{proof}

We now come to the complexity analysis for \textsc{TopNT}. For the shearing step, we remark that there exist at most $n^{O(1)}$ many bad shearing factors $s$ for which our algorithm does not succeed; see~\cite[Thm.~5]{beks:top2D} and \cite[Prop.~11.23]{bpr-arag-06}. Thus, when choosing $s$ at random, we can assume that we succeed for an integer $s$ of bitsize $O(\log n)$. It follows that the sheared polynomial $f(x+sy,y)$ has magnitude $(n,O(\tau+n\log n))$. Hence, throughout the following considerations, we can assume that the leading coefficient $\operatorname{lcf}(f(x,y);y)$ of $f$ (with respect to $y$) is an integer constant and that $f$ has magnitude $(n,O(\tau+n\log n))$. We further define $2^t$ to be a power of two with $\operatorname{lcf}(f(x,y);y)\le 2^t\le 4\cdot\operatorname{lcf}(f(x,y);y)$. 	
	
\begin{lemma}\label{lem:resultants}
		We can compute the entire subresultant sequence $\operatorname{Sres}_i(f, f_y; y)$, with $i=0,\ldots,n$, the polynomial $Q=\operatorname{res}(f_{x}^{*}, f_y^{*}; y)$, and the square-free parts $R^{*}$ and $Q^{*}$ of the corresponding polynomials $R=\operatorname{Sres}_0(f, f_y; y)=\operatorname{res}(f,f_{y};y)$ and $Q$ with $ \tO(n^6 +n^5\tau) $ bit operations.\footnote{The implementation from~\cite{beks:top2D} does not compute the entire sub resultant sequence but only the resultants $R^*$ and $Q^*$. This does not yield any improvement with respect to worst case bit complexity, however, a crucial speed up in practice can be observed.}
	\end{lemma}
	\begin{proof}
		For two bivariate polynomials $g,h\in\Z[x,y]$ of magnitude $(N,\mu)$, computing the subresultant sequence $\operatorname{Sres}_i(g, h; y)\in\Z[x,y]$ together with the corresponding cofactor representations (i.e. the polynomials $u_{i},v_{i}\in\Z[x,y]$ with $u_{i}g+v_{i}h=\operatorname{Sres}_i(g, h; y)$) needs $\tO(N^{5}\mu)$ bit operations~\cite{reischert-asymptotic,Diochnos:2009}. The total degree of the polynomials $\operatorname{Sres}_i(f, f_y; y)$ is bounded by $N^2$, the $y$-degree is bounded by $N-i$, and all coefficients have bitsize $\tO(N\mu)$. 
		Furthermore, according to~\cite[\S 11.2]{Gathen-Gerhard:book}, computing the square-free part of a univariate polynomial of magnitude $(N,\mu)$ uses $\tO(N^{2}\mu)$ bit operations, and the coefficients of the square-free part have bitsize $O(N+\mu)$. Hence, the claim concerning the computation of the polynomials $\operatorname{Sres}_i(f, f_y; y)$ and $R^{*}$ follows from the fact that $f$ and $f_{y}$ have magnitude $(n,O(\tau+n\log n))$ and $R$ has magnitude $(n^{2},\tO(n^2+n\tau))$.

For the computation of the polynomials $f_x^*=\frac{f_x}{\gcd(f_x,f_y)}$ and $f_y^*=\frac{f_y}{\gcd(f_x,f_y)}$, Lemma~\ref{lem:gcdbivariate} yields the bit complexity bound $\tilde{O}(n^6+n^5\tau)$. Lemma~\ref{lem:complexitydiv} implies that $f_x^*$ and $f_y^*$ have magnitude $(n,O(n+\tau))$, and thus computing $Q$ needs $\tilde{O}(n^6+n^5\tau)$ bit operations. $Q$ has magnitude $(n^2,\tO(n^2+n\tau))$.		
\end{proof}
	
We now bound the cost for computing and comparing the roots of $R$ and $Q$.
\begin{lemma}\label{lem:Kplus}
The roots of the polynomials $R$ and $Q$ can be computed with $\tilde{O}(n^{6}+n^{5}\tau)$ bit operations. The same bound also applies to the number of bit operations that are needed to compute the multiplicities $\operatorname{mult}(\alpha,R)$ and $\operatorname{mult}(\alpha,Q)$, where $\alpha$ is a root of $R$.
\end{lemma}

\begin{proof}
According to Theorem~\ref{thm:square free complexity}, we can compute isolating disks for the roots of the polynomials $R$ and $Q$ together with the corresponding multiplicities with $\tO(n^{6}+n^{5}\tau)$ bit operations since $R$ and $Q$ have magnitude $(n^{2},\tilde{O}(n^{2}+n\tau))$. For each root $\alpha$ of $R$, the algorithm returns a disk $\Delta^{(R)}(\alpha):=\Delta(\tilde{\alpha},r_{\alpha})$ with radius $r_{\alpha}<\frac{\sigma(\alpha,R)}{64\deg R}$, and thus we can distinguish between real and non-real roots. A corresponding result also holds for each root $\beta$ of $Q$, that is, each $\beta$ is isolated by a disk $\Delta^{(Q)}(\beta)$ with radius less than $\frac{\sigma(\beta,Q)}{64\deg Q}$. Furthermore, for any given positive integer $\kappa$, we can further refine all isolating disks to a size of less than $2^{-\kappa}$ with $\tO(n^{6}+n^{5}\tau+n^{2}\kappa)$ bit operations. 

For computing the multiplicities $\operatorname{mult}(\alpha,Q)$, where $\alpha$ is a root of $R$, we have to determine the common roots of $R$ and $Q$. This can be achieved as follows: We first compute $d:=\deg\gcd(R^{*},Q^{*})$ for which we need $\tO(n^{6}+n^{5}\tau)$ bit operations. Namely, computing the $\gcd$ of two integer polynomials of magnitude $(N,\mu)$ needs $\tO(N^{2}\mu)$ bit operations. We conclude that $R$ and $Q$ have exactly $d$ distinct roots in common. Hence, in the next step, we refine the isolating disks for $R$ and $Q$ until there are exactly $d$ pairs $(\Delta^{(R)}(\alpha),\Delta^{(Q)}(\beta))$ of isolating disks that overlap. Since $P:=R\cdot Q$ has magnitude $(2n^{2},\tilde{O}(n^{2}+n\tau))$, the minimal distance between two distinct roots $\alpha$ and $\beta$ is bounded by the separation of $P$, thus it is bounded by $2^{-\tilde{O}(n^{4}+n^{3}\tau)}$. We conclude that it suffices to refine the isolating disks to a size of $2^{-\tilde{O}(n^{4}+n^{3}\tau)}$, hence the cost for the refinement is again bounded by $\tO(n^{6}+n^{5})$. Now, for each of the $d$ pairs $(\Delta^{(R)}(\alpha),\Delta^{(Q)}(\beta))$ of overlapping disks, we must have $\alpha=\beta$, and these are exactly the common roots of $R$ and $Q$.
\end{proof}

From the above Lemma, we conclude that we can compute the numbers $k^{+}(\alpha)$ for all roots $\alpha$ of $R$ with $\tilde{O}(n^{6}+n^{5}\tau)$ bit operations. Thus, the same bounds also applies to the computation of the upper bound $K^{+}=\sum_{\alpha}k^{+}(\alpha)$ for $K=\sum_{\alpha} k(\alpha)$.\footnote{For simplicity, we ignored that (in practice) $K^{+}$ can be computed much faster from the equality $K^{+}=n\cdot \deg R^{*}-\deg R+\deg\gcd(R^{\infty},Q)$ instead of computing the $k^{+}(\alpha)$ first and, then, summing up all values.} For the computation of the lower bound $K^{-}$, we use the following result:

\begin{lemma}\label{K-minus}
We can compute $K^{-}$ with $\tilde{O}(n^6+n^{5}\tau)$ bit operations. 
\end{lemma}

\begin{proof}
From the proof of Lemma~\ref{lem:resultants}, we can assume that the leading coefficients $\sr_{i}\in\Z[x]$ of the subresultant sequence 
$\operatorname{Sres}_i(f, f_y; y)$ and the square-free part $S_{0}:=R^{*}$ of the resultant polynomial 
$R$ are already computed. Note that all polynomials $S_{i}$ and $R_{i}$ as defined in (\ref{def:Ris}) have coefficients of 
bitsize $\tO(n\tau+n^{2})$ because all of them divide $R^{*}$. Thus, except for $(n\tau)^{O(1)}$ many bad primes, the modular computation over $\mathbb{Z}_{p}$ 
yields polynomials $S_{i}^{(p)},R_{i}^{(p)}\in\mathbb{Z}_{p}[x]$ with $\deg R_{i}^{(p)}=\deg R_{i}$ for all $i$, and thus $K^{-}=K$. Hence, we can assume that we only have to consider primes $p$ of bitsize $O(\log(n\tau))$.
Since we can compute the polynomials $\sr_{i}$ and $R^{*}$ with $\tilde{O}(n^6+n^{5}\tau)$ bit operations, the same bound also applies to their modular computation over $\Z_{p}$.\footnote{We remark that, in practice, we never compute the entire subresultant sequence over $\Z$. Here, we only assumed their exact computation in order to keep the argument simple and because of the fact that our overall complexity bound is not affected.} 

For the computation of the polynomials $S_{i}^{(p)}\in\Z_{p}[x]$, we have to perform at most $n$ $\gcd$ computations (over $\Z_{p}$ with $p$ of bit size $O(\log (n\tau))$) involving polynomials of degree $n^{2}$. Thus, the cost for these computations is bounded by $\tilde{O}(n\cdot n^{2}\log (n\tau))$ bit operations since computing the $\gcd$ of two polynomials in $\Z_{p}[x]$ of degree $N$ can be achieved with $\tilde{O}(N)$ arithmetic operations in $\Z_{p}$ due to~\cite[Prop.~11.6]{Gathen-Gerhard:book}.
For the computation of the $R_{i}^{(p)}$'s, we have to consider the cost for at most $n$ (remainder-free) polynomial divisions. Again, for the latter computations, we need $\tilde{O}(n\cdot n^{2}\log(n\tau))$ bit operations. 
\end{proof}

We remark that it is even possible to compute $K$ directly in an \emph{expected} number of bit operations bounded by $\tO(n^6+n^{5}\tau)$. Namely, following a randomized approach, the computation of the $\gcd$ of two integer polynomials of magnitude $(N,\mu)$ needs an expected number of bit operations bounded by $\tO(N^{2}+N\mu)$ according to~\cite[Prop.~11.11]{Gathen-Gerhard:book}. This yields the bound $\tO(n(n^{4}+n^{3}\tau))$ for the expected number of bit operations to compute the polynomials $S_{i}$ from the subresultant sequence $\operatorname{Sres}_i(f, f_y; y)$ and the polynomial $S_{0}=R^{*}$. Obviously, the same bound also applies to the computation of the $R_{i}$'s.\\

For the analysis of the curve topology algorithm, it remains to bound the cost for isolating the roots of the ``fiber polynomials'' $f_{\alpha_{i}}(y)\in\R[x]$ and $f_{\beta_{i}}(y)\in\R[x]$, where the $\alpha_{i}$'s are the real roots of $R$ and the $\beta_{i}$'s are arbitrary separating values in between. 
In practice, we recommend to choose 
arbitrary rational values $\beta_{i}$, however, following this straight forward approach yields a bit complexity of 
$\tilde{O}(n^{7}+n^{6}\tau)$ for isolating the roots of the polynomials $f_{\beta_{i}}(y)\in\mathbb{Q}[y]$. Namely, if 
$\beta_{i}$ is a rational value of bitsize $L_{i}$, then $f_{\beta_{i}}$ has coefficients of bitsize 
$\tilde{O}(nL_{i}+\tau)$. Thus, isolating the roots of $f_{\beta_{i}}$ needs 
$\tilde{O}(n^{3}L_{i}+n^{2}\tau)$ bit operations. However, since the separations of the $\alpha_{i}$'s are lower bounded by $2^{-\tilde{O}(n^{4}+n^{3}\tau)}$, we cannot get anything better than $\tilde{O}(n^{4}+n^{3}\tau)$ for the largest $L_{i}$. 

The crucial idea to improve upon the latter approach is to consider, for the values $\beta_{i}$, real roots of the polynomial $\hat{R}(x)$ instead, where $$\hat{R}:=\frac{(R^{*})'}{\gcd((R^{*})',(R^{*})'')}$$ is defined as the square-free part of the derivative of $R^{*}$. Notice that the polynomials $\hat{R}$ and $R$ do not share a common root. Furthermore, from the mean value theorem, we conclude that, for any two consecutive real roots of $R$, there exists a root of $\hat{R}$ in between these two roots.
We can obtain such separating roots by computing isolating disks for all complex roots of $\hat{R}$ such that none of these disks intersects any of the isolating disks for the roots of $R$. The computation of $\hat{R}$ needs $\tO(n^{6}+n^{5}\tau)$ bit operations since $(R^{*})'$ has magnitude $(n^{2},\tilde{O}(n^{2}+n\tau))$. We can use the same argument as in the proof of Lemma~\ref{lem:Kplus} to show that it suffices to compute isolating disks for $R$ and $\hat{R}$ of size $2^{-\tilde{O}(n^{4}+n^{3}\tau)}$ in order to guarantee that the disks do not overlap. Again, Theorem~\ref{thm:refine complexity} shows that we achieve this with $\tilde{O}(n^{6}+n^{5}\tau)$ bit operations. 

Now, throughout the following considerations, we assume that the separating elements $\beta_{i}$ are real roots of $\hat{R}$ with $\beta_{i-1}<\alpha_{i}<\beta_{i}$. We will show in Lemma~\ref{lem:fibers} that, for isolating the roots of all polynomials $f_{\beta_{i}}$ and $f_{\alpha_{i}}$, we need only $\tilde{O}(n^{6}+n^{5}\tau)$ bit operations. For this, we use the following result:

\begin{lemma}\label{mahlermeasure-lemma}
Let $G\in\Z[x]$ be a polynomial of magnitude $(N,\mu)$. For an arbitrary subset $V'\subset \mathcal{V}(G)$, it holds that
$$\sum_{\alpha\in V'} \log \Mea(f_{\alpha}) = \tilde{O}(N\tau+n\mu+N n),\text{ and}\quad
\sum_{\alpha\in V'} \tau_{\alpha}=\tilde{O}(N\tau+n\mu+N n).
$$
In particular, for $G\in\{R,\hat{R}\}$, the bound writes as $\tilde{O}(n^{3}+n^{2}\tau)$.
\end{lemma}
\begin{proof}
The proof is almost identical to the proof of Lemma 5 in~\cite{ks-top-12}. The only difference is that we consider a general $G$, whereas in~\cite{ks-top-12}, only the case $G=R$ has been treated.
Note that $\Mea_\alpha\geq 1$ for every $\alpha\in V(G)$, 
and that the Mahler measure is multiplicative, 
that means, $\Mea(g)\Mea(h)=\Mea(gh)$ for arbitrary univariate polynomials
$g$ and $h$.
Therefore,
$$\sum_{\alpha\in V'}\log\Mea(f_{\alpha})\leq \sum_{\alpha\in\mathcal{V}(G)}\log\Mea(f_{\alpha})=
\log\Mea\left(\prod_{\alpha\in \mathcal{V}(G)}f_{\alpha}\right).$$
Considering $f$ as a polynomial in $x$ with coefficients in $\Z[y]$ yields
$$\prod_{\alpha\in \mathcal{V}(G)}f_{\alpha}=\frac{\res(f,G;x)}{\lcf(G)^n}\Rightarrow
\sum_{\alpha\in V'}\log\Mea(f_{\alpha}) \leq \log\Mea(\res(f,G;x)).$$
It is left to bound the degree and the bitsize of $\res(f,G;x)$.
Considering the Sylvester matrix of $f$ and $G$ (whose determinant defines
$\res(f,G;x)$), we observe that it has $n$ rows with coefficients of $G$
(which are integers of size $O(\mu)$) 
and $N$ rows with coefficients of $f$ (which are univariate polynomials
of magnitude $(n,O(\tau+n\log n))$). Therefore,
the $y$-degree of $\res(f,G;y)$ is bounded by
$O(nN)$, and its bitsize is bounded by $O(n(\mu+\log n)+N(\tau+n\log n))=\tilde{O}(N\tau+n\mu+N n)$.
This shows that $\log\Mea(\res(f,G;x))=\tilde{O}(N\tau+n\mu+Nn)$, and thus the first claim follows. 

For the second claim, note that the absolute value of each coefficient of $f_{\alpha}(y)$ is bounded by $(n+1)\cdot \lambda M(\alpha)^{n}$, where $\lambda=2^{O(\tau+n\log n)}$ is an upper bound for the absolute values of the coefficients of $f$. Thus, we have 
\begin{align*}
\sum_{\alpha\in V'}\tau_{\alpha}&\le \sum_{\alpha\in \mathcal{V}(G)}\tau_{\alpha}\le  \sum_{\alpha\in \mathcal{V}(G)}\log((n+1)\lambda M(\alpha)^{n})\\
&=O(N(\tau+n\log n)+n\log \Mea(G)=\tO(N\tau+n\mu+Nn)
\end{align*}

For the last claim, note that, for $G\in\{R,\hat{R}\}$, we have $N\le n^{2}$ and $\mu=\tO(n^2+n\tau)$.
\end{proof}
	\begin{lemma}\label{lem:bounds}
		For $G\in\{R, \hat{R} \}$, we have
					\begin{align*}
		 \sum_{\alpha\in V(G)} \sum_{i=1}^{k(\alpha)} \mai\log M(\sigmaai^{-1}) &= \tO(n^4+n^3\tau)\quad\text{and}\quad		 \sum_{\alpha\in V(G)} \sum_{i=1}^{\ka}\log M(\Pai^{-1}) = \tO(n^4 +n^3\tau) .
	\end{align*}
	\end{lemma}
	\begin{proof}
		First, consider $G=R$. For any root $\alpha$ of $R$, we define $m(\alpha):=\mult(\alpha,R)$. From (\ref{eq:teissier}), we conclude that $\sum_{i=1}^{k(\alpha)} (m_{\alpha,i}-1) = n-k(\alpha)\le n-\ka+\mult(Q,\alpha)\le\mult(\alpha,R)$. Furthermore, since $m(\alpha)\ge 1$, it follows that  $\mai\le m(\alpha)+1\le 2m(\alpha)$ for all $i$. Hence, we get
		\begin{align*}
			&\sum_{\alpha\in V(R)} \sum_{i=1}^{k(\alpha)} \mai\log M(\sigmaai^{-1})\le
								  \sum_{\alpha\in V(R)} 2\cdot\ma\cdot\sum_{i=1}^{\ka} \log M(\sigmaai^{-1}) \\
								  &\overset{(1)}{=} \tO \left(\sum_{\alpha\in V(R)} \ma\cdot \left(n \log\Mea(\fa) + \log M(\sr_{n-\ka}(\alpha)^{-1})\right)\right) \\
								  &\overset{(2)}{=} \tO \left(\sum_{\alpha\in\mathcal{V}(R)} n \log\Mea(\fa) + \log M(\sr_{n-\ka}(\alpha)^{-1}) \right) \\
								  &\overset{(3)}{=} \tO(n^4+n^3\tau).
		\end{align*}
		For (1), we used~\cite[
 9]{ks-top-12} to show that
		\begin{align}\label{sumoversigma}
		\sum_{i=1}^{\ka} \log M(\sigmaai^{-1}) = \tO(n\log\Mea(\fa)+\log M(\sr_{n-\ka}(\alpha)^{-1}))\ .
		\end{align}\ignore{(recall that $M(x)=\max(1,|x|)$, and hence the sum over all $\log M(\sigmaai^{-1})$ can be interpreted as the sum over the subset of all $\alpha$ with $\sigmaai^{-1}>1$).}
		(2) follows from the fact that each $\alpha\in V(R)$ occurs $\ma$ times in $\mathcal{V}(R)$. Finally, for (3), we apply Lemma~\ref{mahlermeasure-lemma} to bound the first sum and \cite[Lemma 8]{ks-top-12} to bound the second one.
		
		The second claim can be shown as follows. For each $\alpha$, we first split the sum			\begin{align}\label{sumsplit} \sum_{i=1}^{\ka}\log M(\Pai^{-1}) = \sum_{i=1}^{\ka}\log |\Pai|^{-1} + \sum_{i; |\Pai|>1}\log |\Pai|.\end{align}
		Then, for the first sum, we have
		\begin{align}
			&\sum_{i=1}^{\ka}\log |\Pai|^{-1}= \sum_{i=1}^{\ka} \log \prod_{j\neq i}|z_{\alpha,i}-z_{\alpha,j}|^{-m_{\alpha,j}}\nonumber \\
					&=\log\left(\prod_{i=1}^{\ka}\prod_{j\neq i}|z_{\alpha,i}-z_{\alpha,j}|\right)^{-1}+\sum_{i=1}^{k(\alpha)}\log \prod_{j\neq i}|z_{\alpha,i}-z_{\alpha,j}|^{-(m_{\alpha,j}-1)}\nonumber \\
					&\overset{(1)}{\le} \log \frac{|\mathrm{lcf}(f_{\alpha})^{2k(\alpha)-2}|\cdot\prod_{i=1}^{k(\alpha)}m_{\alpha,i}}{|\sr_{n-\ka}(\alpha)|} + \sum_{i=1}^{\ka}\sum_{j\neq i}\log \sigma_{\alpha,i}^{-(m_{\alpha,j}-1)}\nonumber \\
					&\overset{(2)}{=}  \tO(n\tau)+\log |\sr_{n-\ka}(\alpha)|^{-1}+\sum_{i=1}^{\ka}\log \sigmaai^{-1}\sum_{j\neq i}(m_{\alpha,j}-1)\nonumber\\
					&\overset{(3)}{\le}\tO(n\tau)+\log |\sr_{n-k(\alpha)}|^{-1}+\ma\cdot\sum_{i=1}^{k(\alpha)}\log M(\sigma_{\alpha,i}^{-1})\nonumber \\
					&\overset{(4)}{=}\tO(n\tau)+(\ma+1)\cdot \left(n\log\Mea f_{\alpha}+\log M(\sr_{n-k(\alpha)})^{-1}\right)\label{resultfirstsum}
		\end{align}
		For (1), we have rewritten the product as a subresultant term, where we used~\cite[Prop. 4.28]{bpr-arag-06}. Furthermore, the distances $|z_{\alpha,i}-z_{\alpha,j}|$ have been lower bounded by the separation of $z_{\alpha,i}$. For (2), note that $\mathrm{lcf}_{\alpha}$ is an integer of bitsize $O(\tau+\log n)$, that $k\le n$, and that $\prod_{i}\mai\le n^{n}$. For (3), we used that $\sum_{j=1}^{k(\alpha)} (m_{\alpha,j}-1) \le \ma$, and, in (4), we applied (\ref{sumoversigma}). Now, summing up the expression in (\ref{resultfirstsum}) over all $\alpha\in V(R)$ yields
		\[
		\sum_{\alpha\in V(R)} \sum_{i=1}^{\ka}\log |\Pai|^{-1}=\tO(n^{4}+n^{3}\tau),
		\]
		where we again use Lemma~\ref{mahlermeasure-lemma} and \cite[Lemma 8]{ks-top-12}.

		For the second sum in (\ref{sumsplit}), we use that (cf. proof of Theorem \ref{thm:square free complexity})
		$$
			|\Pai| = \frac{|\fa^{(\mai)}(z_{\alpha, i})|}{\mai ! \lcf(f_{\alpha})} < n 2^{\tau_\alpha +1}M(z_{\alpha, i})^{n},
		$$
		and thus
		\begin{align*}
			\sum_{\alpha\in V(R)} \sum_{i; |\Pai|>1}\log |\Pai| &= \tO\left(\sum_{\alpha\in V(R)} \left(n\taua + n\log\prod_{i=1}^{\ka} M(z_{\alpha, i})\right) \right) \\
			&= \tO\left(\sum_{\alpha\in V(R)}\left( n\taua + n\log \Mea(f_\alpha)\right) \right) 
			=\tO(n^{4}+n^3\tau)
		\end{align*}
		according to Lemma~\ref{lem:bounds}. We conclude that $$\sum_{\alpha\in V(R)} \sum_{i=1}^{\ka}\log M(\Pai^{-1}) = \tO(n^4+n^3\tau).$$

		Now, consider the case $G=\hat{R}$. Note that, for each $\alpha\in V(\hat{R})$, we have $\mai=1$ for all $i$, and thus $k(\alpha)=n$. Namely, $R$ and $\hat{R}$ do not share a common root, and thus each polynomial $f_{\alpha}$ has only simple roots. Also, $V(\hat{R})=\mathcal{V}(\hat{R})$ since $\hat{R}$ is square-free. The following computation now shows the first claim
		\begin{align*}
			\sum_{\alpha\in V(\hat{R})} \sum_{i=1}^{\ka} \mai \log M(\sigmaai^{-1})
								  &= \sum_{\alpha\in V(\hat{R})}\left( \sum_{i=1}^{n} \log M(\sigmaai^{-1}) \right)
								  {=} \tO \left(\sum_{\alpha\in V(\hat{R})} n \log\Mea(\fa) + \log M(\sr_{0}(\alpha)^{-1})\right) \\
								  &{=} \tO \left(n^4+n^3\tau + \sum_{\alpha\in V(\hat{R})} \log M(R(\alpha)^{-1}) \right). 
		\end{align*}
		In order to bound the sum in the above expression, note that 
		$$\sum_{\alpha\in V(\hat{R})} \log M(R(\alpha)^{-1}) = \sum_{\alpha\in V(\hat{R})} \log |R(\alpha)|^{-1} + \sum_{\alpha; |R(\alpha)|>1} |R(\alpha)|.$$
		We first compute an upper bound for each value $|R(\alpha)|$. Since $R(x)$ has magnitude $(n^{2},\tO(n\tau))$, it follows that $|R(\alpha)|$ has absolute value less than $2^{\tO(n\tau)}\cdot M(\alpha)^{n^{2}}$. Hence, for any subset $V'\subseteq V(\hat{R})$, it follows that 
		$$\sum_{\alpha\in V'}\log |R(\alpha)|\le \tO(n^{3}\tau)+n^{2} \log \Mea(\hat{R})=\tO(n^{4}+n^{3}\tau).$$
		Thus, it is left to show that $\sum_{\alpha} \log |R(\alpha)|^{-1}=\tO(n^{4}+n^{3}\tau)$, which follows from
		\begin{align}
			\sum_{\alpha\in V(\hat{R})} \log |R(\alpha)|^{-1} &= \log \prod_{\alpha\in V(\hat{R})} |R(\alpha)|^{-1}= \log\left(\frac{|\lcf(\hat{R})|^{\deg(R)}}{|\res(R,\hat{R})|}\right) = \tO(n^4 + n^3\tau)\ .\label{eq:Rbound}
		\end{align}

		In the second equation we rewrote the product in terms of the resultant $\res(R,\hat{R})$ \cite[Prop. 4.16]{bpr-arag-06}. Since $R$ and $\hat{R}$ have no common root, we have $|\res(R,\hat{R})|\ge 1$. Thus, the last equation follows from the fact that the leading coefficient of $\hat{R}$ has bitsize $\tO(n^2+n\tau)$ and that $\deg(R)\le n^2$.

		Similarly, for the second claim, we first derive an upper bound for 
		$\sum_{\alpha\in V(\hat{R})} \sum_{i:\Pai>1}\log |\Pai|$. Again, we can use exactly the 
		same argument as for the case $G=R$ to show that the latter sum is bounded by 
		$\tO(n^{4}+n^{3}\tau)$. Hence, it suffices to prove that 
		$$\sum_{\alpha\in V(\hat{R})} \sum_{i=1}^{\ka}\log |\Pai|^{-1}=\tO(n^{4}+n^{3}\tau).$$ which follows from
		\begin{align*}
			\sum_{\alpha\in V(\hat{R})} \sum_{i=1}^{\ka}\log |\Pai|^{-1} &= \sum_{\alpha \in V(\hat{R})}\sum_{i=1}^{n} -\log \prod_{j\neq i}|z_{\alpha,i}-z_{\alpha,j}| {=} \sum_{\alpha \in V(\hat{R})}\left( -\log \prod_{i=1}^{n}\prod_{j\neq i}|z_{\alpha,i}-z_{\alpha,j}|  \right) \\
					&{=} \sum_{\alpha \in V(\hat{R})}\left( -\log \frac{|\sr_{0}(\alpha)|}{|\mathrm{lcf}(f_{\alpha})^{2n-2}|}\right) {=} \sum_{\alpha \in V(\hat{R})}\left( -\log \frac{|R(\alpha)|}{|\mathrm{lcf}(f_{\alpha})^{2n-2}|}\right) = \tO(n^4+n^3\tau)\ .
		\end{align*}
		The last step follows from (\ref{eq:Rbound}). We remark that the above computation is similar to the one for the case $G=R$. However, we used the fact that $\hat{R}$ is square-free, and thus all multiplicities $\mai$ are equal to one.
	\end{proof}

\begin{lemma}\label{lem:resultant isolation}
	Let $G\in\{R,\hat{R}\}$ and let $L_{\alpha}\in\N$ be arbitrary positive integers, where $\alpha$ runs over all real roots of $G$. Then, we can compute an absolute $L_{\alpha}$ approximations for all polynomials $f(\alpha,y)$ using
	$$\tO(n^6+n^5\tau+n^{2}\sum_{\alpha}L_{\alpha})$$
	bit operations.
\end{lemma}
\begin{proof}
	For each $\alpha$, we use 
	approximate interval arithmetic to compute an approximation of the polynomial $f_{\alpha}$. If we 
	choose a fixed point precision $\rho$, and a starting interval of size $2^{-\rho}$ that contains 
	$\alpha$, then the so-obtained interval approximation of $f_{\alpha}$ has interval coefficients of 
	size $2^{-\rho+2}(n+1)^{2}2^{\tau}M(\alpha)^{n}$; 
	see again~\cite[Section 4]{qir-kerber-11} and~\cite[Section 5]{ks-top-12} for more details. Thus, in order to get an approximation of 
	precision $L_{\alpha}$ of $f_{\alpha}$, it suffices to consider a $\rho$ of size 
	$\tO(\tau+n\log M(\alpha)+L_{\alpha})$. Thus, by doubling the precision $\rho$ in each step, we eventually succeed for some $\rho=\rho_{\alpha}=\tO(\tau+n\log M(\alpha)+L_{\alpha})$. The cost for the interval evaluations is then dominated (up to a logarithmic factor) by the cost in the last iteration. Thus, for a certain $\alpha$, the cost is bounded by $\tO(n^{2}(\tau+n\log M(\alpha)+L_{\alpha}))$ since for each of the $n+1$ coefficients of $f_{\alpha}$, we have to (approximately) evaluate an integer polynomial (i.e.~the coefficients of $f$ considered as a polynomial in $y$) of magnitude $(n,O(\tau+\log n))$ at $x=\alpha$. The total cost for all $\alpha$ is then bounded by
	\[
	\tO\left(n^{2}\cdot\sum_{\alpha\in\R\cap V(G)} \tau+n\log M(\alpha)+L_{\alpha}\right)=\tO(n^{6}+n^{5}\tau+n^{2}\sum_{\alpha}L_{\alpha}),
	\]
	where we again used the result in Lemma~\ref{lem:bounds}. For the interval evaluations, we need an approximation of the root $\alpha$ to an absolute error of less than $2^{-\rho_{\alpha}}$. Such approximations are provided if we compute isolating disks of size less than $2^{-\kappa}$ for all roots of $G$, given that $\kappa$ is larger than $\max_{\alpha} \rho_{\alpha}=\tO(\tau+n\max_{\alpha}\log M(\alpha)+ \max_{\alpha} L_{\alpha})=\tO(n^{3}+n^{2}\tau+\max_{\alpha}L_{\alpha})$. In the proof of Lemma \ref{lem:Kplus}, we have already shown that we can compute such disks using $\tO(n^6+n^5\tau+n^{2}\kappa)$ bit operations. Thus, the claim follows. 
\end{proof}

	\begin{lemma}\label{lem:fibers}
	Let $G\in\{R,\hat{R}\}$. Then, computing isolating disks for all roots of all $f_{\alpha}$, $\alpha\in V(G)\cap\R$, together with the corresponding multiplicities uses $$ \tO(n^6+n^5\tau) $$ bit operations. 
	\end{lemma}
	\begin{proof}
		\ignore{Assume that all polynomials $f_{\alpha}$ are already given in form of approximations of precision $L$ with sufficiently large $L$. 
		We first analyze the amortized bit complexity of applying the root isolation algorithm from Section \ref{sec:rootisolation} to all polynomials. 
		From these calculations we then determine how well the polynomials $f_{\alpha}$'s have to be approximated, i.e. we give a choice for $L$ that is sufficiently large and show the overall bit complexity.}
	
		For a fixed $\alpha$, let $B_\alpha$ be the number of bit operations that are needed to compute isolating disks for all roots of $\fa$ together with the corresponding multiplicities. We apply Theorem \ref{thm:isolation complexity} to $2^{-t}\cdot f_{\alpha}$ which has exactly the same roots as $f_{\alpha}$ (remember that $\operatorname{lcf}(f(x,y);y)\in\Z$ and $\operatorname{lcf}(f(x,y);y)\le 2^t\le 4\cdot \operatorname{lcf}(f(x,y);y)$). Then, we have
$$ B_\alpha = \tO\left(n^3 + n^2\taua + n\cdot\sum\nolimits_{i=1}^{\ka}\left( \mai\log M(\sigmaai^{-1}) + \log M(\Pai^{-1}) \right)\right).$$ The corresponding algorithm from Section~\ref{sec:algorithm} returns isolating disks for the roots $z_{\alpha,i}$ and their multiplicities $m_{\alpha,i}$. Furthermore, since the radius of the disk isolating $z_{\alpha,i}$ is smaller than $\sigmaai/(64n)$, we can distinguish between real and non-real roots. The algorithm needs an absolute $L_{\alpha}$-approximation of $2^{-t}\cdot f_{\alpha}$ (and thus an absolute $(L_{\alpha}-t)$-approximation of $f_{\alpha}$) with 
\[
L_{\alpha}=\tO\left(n\tau_{\alpha} + \sum\nolimits_{i=1}^{\ka}\left( \mai\log M(\sigmaai^{-1}) + \log M(\Pai^{-1})\right)\right).
\]
From Lemma~\ref{lem:resultant isolation}, we conclude that we can compute corresponding approximations for all $f_{\alpha}$, $\alpha\in V(G)\cap\R$, with a number of bit operations bounded by
\[
\tO(n^{6}+n^{5}\tau+n^{2}\cdot\sum_{\alpha\in V(G)\cap\R} L_{\alpha}).
\]
The above expression is bounded by $\tO(n^{6}+n^{5}\tau)$ because
\[
\sum_{\alpha\in V(G)} \left(n\tau_{\alpha} + \sum\nolimits_{i=1}^{\ka}\left( \mai\log M(\sigmaai^{-1}) + \log M(\Pai^{-1})\right)\right)
\]
is bounded by $\tO(n^{4}+n^{3}\tau)$ according to Lemma~\ref{mahlermeasure-lemma} and~\ref{lem:bounds}.
The same argument also shows that the sum over all $B_{\alpha}$ is even bounded by $\tO(n^{5}+n^{4}\tau)$. Hence, the claim follows.
\ignore{
a comparable bound. Hence, we are done if we can show that
\[
sum_{\alpha\in V(G)\cap\R}
\]

Hence, computing the isolation for all polynomials $f_{\alpha}$, $\alpha\in V(R)$ takes
		$$
			\tO \left(\sum_{\alpha\in V(R)} T_\alpha\right)
		$$
bit operations. 
		Lemma \ref{lem:bounds} lets us bound $$\sum_{\alpha\in V(R)}\sum_{i=1}^{\ka}\left( \mai\log M(\sigmaai^{-1}) + \log M(\Pai^{-1}) \right) = \tO(n^4+n^3\tau) .$$ For the remaining terms, we have
		\begin{align*}
			\sum_{\alpha\in V(R)} n^2\taua &\le n^2 \sum_{\alpha\in V(R)} \taua \\
											 &\overset{(1)}{=} \tO\left(n^2 \sum_{\alpha\in V(R)} ( n\log |a|+\tau)\right) \\
											 &\overset{(2)}{=} \tO\left(n^3 (\tau+\log\Mea(R))\right) \overset{(3)}{=} \tO(n^3\tau).
		\end{align*}
		Since the coefficients of $\fa$ are evaluations of $f$ at $\alpha$, their size is bounded by $\tO(n\log |a|+\tau)$ (1). For (2), we then used the definition of the Mahler measure to rewrite the sum $\sum_{\alpha\in V(R)}\log |a|$. 
		For (3) we bounded $\log\Mea(R)$ by $\tO(\tau)$ \cite[Prop. 10.9]{bpr-arag-06}.
		
		Plugging in the calculations above, we obtain a bound of $
                \tO(n^5+n^4\tau) $ for isolating the roots of all fiber
                polynomials, while requiring approximations of
                precision $L=\tO(n^4+n^3\tau)$ of $\fa$, using the bound given in Theorem \ref{thm:isolation complexity}. \ignore{\Kurt{what are well approximated
                  approximation of precision $L$?}}}
	\end{proof}

\ignore{
\begin{lemma}\label{lem:intermediates}
	Computing intermediate values $\beta_1,...,\ldots,\beta_{k+1}$ and isolating the roots of all $f_{\beta_i}$ for all $1\le i\le k+1$ takes
$$\tO(n^6+n^5\tau)$$
bit operations.
\end{lemma}
\begin{proof}
	Since we pick our intermediate values as the roots of $\hat{R}$, computing them corresponds to computing a root isolation of $\hat{R}$ such that no isolating disk of $\hat{R}$ intersects any of the isolating disks of $R$. Since the roots of $\hat{R}$ are those of the derivative of the square-free part of $R$, the distance between any root $\beta$ of $\hat{R}$ and any root $\alpha$ of $R$ is at least $\sigma(\alpha,R)/\deg(R)$ \cite[Thm. 8]{eigenwillig-roots}. Hence, this is given if we isolate the roots to disks of size $2^{-\tO(n^4+n^3\tau)}$.

Similar to the proof in Lemma \ref{lem:resultant isolation}, and using Theorem \ref{thm:refine complexity}, we have that computing approximations of precision $L=\tO(n^4+n^3\tau)$ of all intermediate polynomials $f_\beta=f(\beta,y)$, $\beta\in V(\hat{R})$ takes $\tO(n^6+n^5\tau)$ bit operations. Using the same argument as in Lemma \ref{lem:fibers}, performing the isolation step for all polynomials $f_{\beta}$ also falls into the bit complexity of $\tO(n^6+n^5\tau)$.

\end{proof}}

We can now formulate our main theorems of this section:

\begin{theorem}\label{thm:complexitycurves}
	Computing the topology of a real planar algebraic curve $C=V(f)$, where $f\in\Z [x,y]$ is a bivariate polynomial of total degree of $n$ with integer coefficients of magnitude bounded by $2^\tau$, needs an expected number of bit operations bounded by
$$ \tO(n^6+n^5\tau).$$
\end{theorem}
\begin{proof}
We already derived a bound of $\tO(n^{6}+n^{5}\tau)$ or better for each of the steps in the projection and in the lifting phase of our algorithm. The final connection phase is purely combinatorial since we ensure that each $f_{\alpha}$, with $\alpha$ a root of the resultant $R$, has at most one multiple real root. Thus, we can compute all adjacencies in linear time with respect to the number of roots of critical and intermediate fiber polynomials. Since their number is bounded by $O(n^{3})$, this step can be done in $O(n^3)$ operations.
\ignore{Lemma \ref{lem:resultants} shows that we can compute the relevant resultants $R$, $\hat{R}$, and $Q$ in the claimed bit complexity. Isolating the roots of $R$ and $\hat{R}$ gives us $x$-critical and intermediate values, which is again possible with $\tO(n^6+n^5\tau)$ bit operations (Lemmas \ref{lem:resultant isolation} and \ref{lem:intermediates}). Determining the number of distinct roots $\ka$ for all $\alpha\in V(R)$ also falls into the same bit complexity, as is shown in Lemmas \ref{lem:Kplus} and \ref{K-minus}. With this information, we can use the root isolator from Section \ref{sec:algorithm} to isolate the roots of the fiber polynomials $f_\alpha$ and $f_\beta$, where $\alpha\in V(R)$, and $\beta\in V(\hat{R})$. Lemmas \ref{lem:fibers} and \ref{lem:intermediates} pin the complexity of this step also to $\tO(n^6+n^5\tau)$ bit operations.}
\end{proof}

\Kurt{Notice that the problem of (real) solving a bivariate polynomial system $g(x,y)=h(x,y)=0$, with $g,h\in\Z[x,y]$ coprime polynomials, can be reduced to the problem of computing the topology of the planar algebraic curve $C:=\{(x,y)\in\R^2:f(x,y)=0\}$, where $f:=g^2+h^2$. Namely, $C$ coincides with the set of points $(x,y)\in\R^2$ for which both polynomials $g$ and $h$ vanish. Since the degree of $f$ is twice as large as the maximum of the degrees of the polynomials $g$ and $h$, the following result follows in an almost straight forward manner:\footnote{We remark that we consider this result of rather theoretical interest because there exist efficient algorithms for solving bivariate systems that are comparable fast (in practice) as the fastest algorithms for topology computation; see~\cite[Section 6]{beks:top2D} for extensive benchmarks. Increasing the degree of the input polynomials by a factor of $2$ certainly does not harm the asymptotic complexity bounds, however, it has a significant impact on the practical running times.}}

\Kurt{\begin{theorem}\label{complexity:systemsolving}
Let $g,h\in\Z[x,y]$ be coprime polynomials of magnitude $(n,\tau)$. Then, we can compute isolating boxes for the real solutions of the system $g(x,y)=h(x,y)=0$ with an expected number of bit operations bounded by $$\tO(n^{6}+n^{5}\tau).$$
\end{theorem}}

\begin{proof} \Kurt{
As already mentioned above, the idea is to consider the polynomial $f(x,y):=g^{2}+h^{2}$ and to compute the topology of the curve $C:=C_{\R}$ defined by $f$. Since $g$ and $h$ are assumed to be coprime, the system $g=h=0$ has only finitely many solutions, and the set of these points coincides with the ``curve'' $C$. Hence, the topology algorithm returns a graph that consists of vertices only. According to Theorem~\ref{thm:complexitycurves}, the cost the topology computation is bounded by $\tO(n^{6}+n^{5}\tau)$ bit operations in expectation since $f$ has magnitude $(2n,O(\tau+\log n))$.}

\Kurt{However, in general, our algorithm does not directly return the solutions of the initial system but the solutions of a sheared system $g(x+sy,y)=h(x+sy,y)=0$. Here, $s$ is a positive integer of bitsize $O(\log n)$ for which \textsc{TopNT} succeeds in computing 
the topology of the sheared curve $\hat{C}:=C_{s,\R}$ defined by $\hat{f}(x,y)=f(x+sy,y)=0$. Since $\hat{C}$ consists of isolated singular points only and there are no two covertical points (note that our algorithm only succeeds for an $s$ for which there are no two covertical extremal points), it follows that, for each point $(\hat{x},\hat{y})\in \hat{C}$, $\hat{x}$ is a root of the resultant $\hat{R}=\res(\hat{f},\hat{f};y)$ and $\hat{y}$ is the unique (multiple) real root of $\hat{f}(\hat{x},y)$. The point $(\hat{x},\hat{y})$ is represented by an isolating box $B(\hat{x},\hat{y})=I(\hat{x})\times I(\hat{y})$, where $I(\hat{x})$ is the isolating interval for the root $\hat{x}$ of $\hat{R}$ and $I(\hat{y})$ is the isolating interval for the root $\hat{y}$ of $f(\hat{x},y)$. Each solution $(x,y)$ of the initial system can now be recovered from a unique solution $(\hat{x},\hat{y})\in\hat{C}$. More precisely, $x=\hat{x}-s\cdot \hat{y}$ and $y=\hat{y}$. However, in order to obtain isolating boxes for the solutions $(x,y)$, we have to refine the boxes $B(\hat{x},\hat{y})$ first such that the sheared boxes $B(x,y):=(I(\hat{x})-s\cdot I(\hat{y}),I(\hat{y}))$ do not overlap. Note that the latter is guaranteed if both intervals $I(\hat{x})$ and $I(\hat{y})$ have width less than $\sigma(\hat{x},\hat{R})/(4|s|)\le \sigma(\hat{x},\hat{R})/4$. Namely, if the latter inequality holds, then the intervals $I(\hat{x})-s\cdot I(\hat{y})$ are pairwise disjoint.
Hence, it follows that the corresponding isolating intervals have to be refined to a width less than $w(\hat{x},\hat{y})=\sigma(\hat{x},\hat{R})/n^{O(1)}$. For the resultant polynomial $\hat{R}$, we conclude from Theorem~\ref{thm:square free complexity} that computing isolating intervals of size less $w(\hat{x},\hat{y})$ uses $\tO(n^{6}+n^{5}\tau)$ bit operations since $\log M(w(\hat{x},\hat{y})^{-1})=\tO(n^{4}+n^{3}\tau)$ and $\hat{R}$ has magnitude $(n^{2},\tO(n\tau))$. In order to compute an isolating interval of size $w(\hat{x},\hat{y})$ or less for the root $\hat{y}$ of $\hat{f}(\hat{x},y)$ (in fact, for all roots of $\hat{f}(\hat{x},y)$), we need 
\begin{align*}\tO(n^3 + n^2\tau_{\hat{x}}& + n\cdot\sum\nolimits_{i=1}^{k(\hat{x})}\left( m_{\hat{x},i}\log M(\sigma_{\hat{x},i}^{-1}) + \log M(P_{\hat{x},i}^{-1}) \right)+(n\max_{i} m_{\hat{x},i})\cdot\log M(w(\hat{x},\hat{y})^{-1})
\end{align*}
bit operations; cf.~the proof of Lemma~\ref{lem:fibers} with $\alpha=\hat{x}$ and $f=\hat{f}$. Also, we need an approximation of precision $L_{\hat{x}}$ of $\hat{f}(\hat{x},y)$ with $L_{\hat{x}}$ bounded by
\begin{align*}
\tO(n\tau_{\hat{x}} &+ \sum\nolimits_{i=1}^{k(\hat{x})}\left( m_{\hat{x},i}\log M(\sigma_{\hat{x},i}^{-1}) + \log M(P_{\hat{x},i}^{-1})\right)+(n\max_{i} m_{\hat{x},i})\cdot\log M(w(\hat{x},\hat{y})^{-1}).
\end{align*} 
Since $\max_{i} m_{\hat{x},i}\le 2\cdot\mult(\hat{x},\hat{R})$ and $w(\hat{x},\hat{y})=\sigma(\hat{x},\hat{R})/n^{O(1)}$, it holds that $$(n\max_{i} m_{\hat{x},i})\cdot\log M(w(\hat{x},\hat{y})^{-1})=\tO(n\cdot\mult(\hat{x},\hat{R})\cdot\log M(\sigma(\hat{x},\hat{R})^{-1}).$$ Thus, summing up the cost for computing the roots of $\hat{f}(\hat{x},y)$ over all real roots of $\hat{R}$ yields the bound $\tO(n^{5}+n^{4}\tau)$. Here, we use an analogous argument as in the proof of Lemma~\ref{lem:fibers} and the fact that $\sum_{\hat{x}}n\cdot\mult(\hat{x},\hat{R})\cdot\log M(\sigma(\hat{x},\hat{R})^{-1}=\tO(n^{5}+n^{4}\tau)$. The more costly part is to compute the approximations of precision $L_{\hat{x}}$ of the polynomials $\hat{f}(\hat{x},y)$. Again, we can use Lemma~\ref{mahlermeasure-lemma} and~\ref{lem:bounds} to show that $\sum_{\hat{x}}L_{\hat{x}}=\tO(n^{4}+n^{3}\tau)$. Thus, from Lemma~\ref{lem:resultant isolation}, we conclude that the approximations of the $\hat{f}(\hat{x},y)$'s can be computed with $\tO(n^{6}+n^{5}\tau)$ bit operations.}
\end{proof}

\section{Conclusion}

We presented an algorithm for isolating the roots of a complex univariate polynomial that can handle multiple roots provided the number $k$ of distinct roots is part of the input and the coefficients can be approximated to an arbitrary precision. The algorithm uses approximate factorization as a subroutine. Any algorithm for approximate factorization that can be run with arbitrary precision can be used.

If used with Pan's algorithm~\cite{Pan:alg} for approximate factorization, the algorithm is highly efficient:\medskip
\begin{itemize}
\item It solves the benchmark problem 
of isolating all roots of a polynomial $p$ with \emph{integer} coefficients of absolute value bounded by
$2^{\tau}$ with 
$\tO(n^3 + n^2 \tau)$ bit operations. This matches the best bound known~\cite[Theorem 3.1]{Pan:survey}. 
\item When combined with a a recent algorithm for computing
  the topology of a real planar algebraic curve specified as the zero set of a bivariate
  integer polynomial, it leads to improved complexity bounds for topology computation and and for isolating the real solutions of a bivariate
  polynomial system. For input polynomials of degree $n$ and bitsize $\tau$, we 
  improve the currently best running
  time from
$\tO(n^{9}\tau+n^{8}\tau^{2})$ (deterministic) to $\tO(n^{6}+n^{5}\tau)$
(randomized) for topology computation and from $\tO(n^{8}+n^{7}\tau)$
(deterministic) to $\tO(n^{6}+n^{5}\tau)$ (randomized) for solving bivariate
systems.\medskip
\end{itemize}

The considerable improvement of the bit complexity of the above problems related to the computation of a cylindrical algebraic decomposition mainly stems from the adaptivity of our root isolation method. That is, the precision demand as well as the number of bit operations to isolate the roots of the "fiber polynomials" is directly related to the geometric locations of the corresponding roots. As a consequence, our analysis profits from amortization effects over the critical fibers. 
We expect that our adaptive complexity bound for root isolation will yield a a series of further
complexity results for similar problems, where amortization effects take place. 

\Kurt{A major open problem is whether there are deterministic algorithms for curve analysis and bivariate system solving of the same complexity.}

\end{document}